\documentclass[11pt]{article}

\usepackage{environment-setup}
\usepackage{macros}
\usepackage[linesnumbered,lined,boxed]{algorithm2e}
\usepackage{nicematrix}

\title{
Polynomial Bounds on Degeneration Order from Commutativity Properties of Tensor Slices}
\author{
Shree Ganesh
\thanks{ENS de Lyon, CNRS, Université Claude Bernard Lyon 1, LIP, UMR 5668, 69342, Lyon cedex 07, France. SG is supported by a MATHINFI PhD scholarship.}
\and 
Pascal Koiran
\thanks{ENS de Lyon, CNRS, Université Claude Bernard Lyon 1, LIP, UMR 5668, 69342, Lyon cedex 07, France.}
\and
Rafael Oliveira~\orcidlink{0000-0001-8917-8689}
\thanks{University of Waterloo. \ \email{rafael@uwaterloo.ca}.
Part of this work took place while RO was supported by the Simons Institute for the Theory of Computing, and conducted when the author was visiting the Institute.
The author would like to thank the generosity, hospitality, and excellent research environment provided by the Simons Institute. }
}
\date{\today}

\begin{document}

\maketitle

\begin{abstract}
A tensor has border rank at most $r$ if it can be written as $T=\lim_{\e \rightarrow 0} T(\e)$ where $T(\e)$ has rank at most $r$ for all sufficiently small $\e$. 
It is known that the map $\e \mapsto T(\e)$ can be assumed to be a (tensor valued) polynomial in $\e$. 
The smallest possible degree of such a map is called the error degree of $T$.
The error degree and the related notion of {\em order of degeneration} are the two key quantities that we study in this paper.
One motivation comes from debordering: by polynomial interpolation on the map  
$\e \mapsto T(\e)$ we can upper bound the tensor rank of $T$.

For order 3 tensors, exponential upper bounds on the error degree and  degeneration order  were given almost 40 years ago in~\cite{LEHMKUHL19891} and were not improved ever since.
In this paper we give bounds that apply to a wide class of tensors, exponentially improving on~\cite{LEHMKUHL19891}. 

Our results are most general for tensors with 3 slices (format $m \times n \times 3$). 
In this case, our main assumption is on the rank of the matrix slices.
We also give bounds that apply to arbitrary rectangular formats ($m \times n \times p$). 
In this case, we need an additional 1-regularity assumption on one of the slices of the tensor (recall that a matrix is said to be 1-regular if its eigenspaces are 1-dimensional). 
Under these assumptions we show that
the error degree is at most 1, which yields
a nontrivial debordering result (tensor rank at most~$2r$ for border rank $r$).

The results in~\cite{LEHMKUHL19891} rely on an upper bound on the degree of the variety of tensors of border rank at most $r$. We rely instead on more specific properties of this algebraic variety, and in particular on commutativity properties of certain matrices derived from the tensor slices.
   
\end{abstract}

\thispagestyle{empty}

\newpage
\tableofcontents
\thispagestyle{empty}

\ifdefined\DRAFT
    \linenumbers
\else 
\fi

%================================================================================
\newpage
\setcounter{page}{1}
\section{Introduction}
%================================================================================

The rank of a 3-tensor $T\in\bC^{m\times n\times p}$, denoted $\Rank(T)$, is the minimum positive integer $r$ such that $T=\sum_{i=1}^r u_i\tp v_i\tp w_i$ for some vectors $u_i\in\bC^m,v_i\in\bC^n,w_i\in\bC^p$. 
Tensor rank is a fundamental quantity, appearing prominently in several areas of science, mathematics and in important questions in complexity theory, since the foundational works of Strassen~\cite{strassen1969gaussian,strassen1973vermeidung} on matrix multiplication.
For a survey of tensor rank in other areas, we refer the reader to \cite{kolda2009tensor}. 
For an account of tensor rank and matrix multiplication, we refer the reader to \cite{BCS97,blaser2013fast}.

The works of Bini et al.~\cite{bini1979n2,bini1980relations}, also showed the importance of approximate notions of tensor rank to the design of fast algorithms for matrix multiplication. 
These notions of approximative algorithms led to the definition of border rank, which we define next, following~\cite[page 379]{BCS97}.
Some equivalent (geometric/topological) definitions of border rank can be found in~\cite[Chapter 20.6]{BCS97} and in~\cite[Section 3.2]{dutta2025recent}.

\begin{definition}[Border Rank]
    The border rank of a tensor $T\in\bC^{m\times n\times p}$, denoted $\Brank(T)$, is the minimum $r$ for which there exist $q\in\bN$, $u_i(\e)\in\bC[\e]^m,v_i(\e)\in\bC[\e]^n,w_i(\e)\in\bC[\e]^p$ and a tensor $Q(\e)\in\bC[\e]^{m\times n\times p}$ such that the following equality holds:
    \begin{linenomath}
        $$\e^{q}T +\e^{q+1}Q(\e)=\sum_{i=1}^r u_i(\e)\tp v_i(\e)\tp w_i(\e)$$
    \end{linenomath}
\end{definition}

% \RO{do we need this remark?}

% \begin{remark}
%     $T(\e):= \frac{1}{\e^q}\sum_{i=1}^r u_i(\e)\tp v_i(\e)\tp w_i(\e)$, where $\sum_{i=1}^r u_i(\e)\tp v_i(\e)\tp w_i(\e)$ is from the definition above, is the tensor approximating $T$. In other words, $lim_{\e\rightarrow 0}T(\e)=T$ where $\Rank(T(\e))\leq r$ for all $\e \neq 0$.
%     \PK{where $\Rank(T(\e)) \leq r$ for all $\epsilon \neq 0$. The converse of this observation also holds true: if $(T_n)$ is a sequence of tensors converging to $T$
%     and $\Rank(T_n) \leq r$ for all $n$, then $\Brank(T) \leq r$. This follows from 
%     the fact that the set of tensors of border rank at most $r$ is the 
%     Zariski closure (and also the Euclidean closure) of the set of tensors of rank 
%     at most $r$: see~\cite[chapter 20.6]{BCS97} and~\cite[section 3.2]{dutta2025recent}.}
% \end{remark}

\noindent
From the above definition of border rank, we can say that any decomposition of the above form is an \say{(optimal) \emph{border decomposition}} for $T$. 
It is also natural to define a quantitative measure of approximation of a tensor, done by the notion of degeneration order \cite[Definition 15.19]{BCS97}.

\begin{definition}[Degeneration Order]
    Let $T\in \bC^{m\times n \times p}$ be a tensor with $\Brank(T) = r$. 
    The \emph{degeneration order} of $T$, denoted $\bdorder(T)$, is the smallest integer $q$ for which there are $Q(\e) \in \bC[\e]^{m\times n \times p}$, $u_i(\epsilon) \in \bC[\e]^m, v_i(\epsilon) \in \bC[\e]^n$ and $w_i(\epsilon) \in  \bC[\e]^p$ satisfying
    \begin{linenomath}
        $$\sum_{i=1}^r u_i(\e)\tp v_i(\e) \tp w_i(\e) = \e^q T + \e^{q+1}Q(\e).$$ 
    \end{linenomath}
\end{definition}

\noindent In algebraic complexity theory, the above notions of rank and border rank were generalized to parametrized classes of algebraic circuits, giving rise to the notion of border complexity classes. 
An important question in the field is to understand the power of such border complexity classes. 
The (proper) debordering question can be phrased as follows: given a circuit class $\cC$, is its border class $\overline{\cC}$ polynomially related to $\cC$?
More generally: what is the relation between a given circuit class and its border class?
Understanding the answer to these questions for certain complexity classes would have important implications in several problems in algebraic complexity, such as in the factoring problem as well as whether GCT's border Valiant hypothesis is equivalent to Valiant's hypothesis.
For more on the importance of debordering in algebraic complexity, we refer the reader to the recent survey \cite{dutta2025recent} and references therein.

In the case of 3-tensors, as any tensor in $\bC^{m \times n \times p}$ has rank $O(\min(mn, mp, np))$, the debordering question asks what is the tightest relation between border rank and tensor rank.
By standard interpolation results, one can relate tensor rank and border rank via the degeneration order to obtain $\Rank(T) \leq (2 \cdot \bdorder(T) + 1) \cdot \Brank(T)$ -- see \cite[Prop 15.26]{BCS97}. 
This gives us one avenue towards \say{debordering} tensor rank.
However, the best known bound on the degeneration order of a tensor in $\bC^{m \times n \times p}$ still is the original bound of $\bdorder(T) \leq 3^{(m+n+p-3) \cdot \Brank(T)}$ obtained by Lehmkuhl and Lickteig \cite{LEHMKUHL19891}.
Thus, if we were to use this upper bound, we would obtain a trivial answer to the debordering question for tensors.

Other notions of approximation which are also useful in the debordering question were introduced in the survey \cite[Section 3.3]{dutta2025recent}, in particular \cite[Definitions 39 and 45]{dutta2025recent}: the notions of approximating curve and error-degree.
In this paper we will consider the following slightly different notion of \emph{error-degree}, which sits in between these two definitions\footnote{Our definition of \textit{error-degree} of a tensor is inspired by \cite[Definitions 39 and 45]{dutta2025recent} where they define a similar notion for a vast class of ``border complexity measures'' (which includes border rank as a special case). 
Their definition of error-degree is more stringent in the sense that they require $\Rank T(\e) \leq r$ for all $\epsilon \neq 0$ whereas we allow finitely many exceptions in \cref{def:errordegree}, much like their definition of approximating curve. Note that our results on the order of degeneration imply bounds on this more stringent version of error degree: see~\cref{lem:errorfromdegeneration} and~\cref{rem:2bdorder}.}, but it is still good enough to obtain debordering results via the standard interpolation techniques (see \cref{lem: interpolation}).

\begin{definition}[Error degree]\label{def:errordegree}
    Let $T \in \bC^{m\times n \times p}$ be such that $\Brank(T) = r$.  
    The \textit{error-degree} of $T$, denoted $\edeg(T)$, is the smallest $k \in \bN$ for which there exist tensors $T_i\in \bC^{m\times n \times p}$ for all $1\leq i \leq k$ such that $T(\e):= T + \e T_1 + \dots + \e^kT_k$ has rank at most $r$ for all but finitely many $\e \in \bC$.
\end{definition}

In the survey \cite[Open question 1]{dutta2025recent}, the authors have asked the question of whether one could prove better bounds on the error degrees for common complexity measures.
In this paper, we make progress on this question for the (border) tensor rank measure for 3-tensors, and we moreover improve the upper bound on the border degeneration for 3-tensors.
We obtain both of these improved results for important special classes of 3-tensors, which we will now discuss before formally stating our results.

\paragraph{Classes of 3-tensors.} A tensor $T \in \bC^{m \times n \times p}$ is \emph{concise} if it cannot be expressed as a tensor in a smaller ambient space.\footnote{For a formal definition, see \cite[Section 1.1]{JLP2024concise}.} 
By \cite[Lemma 15.23]{BCS97}, the border rank of a concise tensor is lower bounded by the maximum of the three given dimensions of the tensor.
Hence, among the concise tensors we have those which are of \emph{minimal border rank}, i.e. those tensors whose border rank equals $\max(m,n,p)$.
As mentioned in \cite[Problem 15.2]{BCS97}, the class of \emph{concise} and \emph{minimal border rank} tensors is an important class of 3-tensors in the study of matrix multiplication.

In the study of the classification of concise tensors of minimal border rank (for background, see \cite{JLP2024concise}), certain important subclasses are considered in order to make the problem a bit more tractable (even though it is known that the classification problem even for such subclasses is equivalent to an extremely challenging problem in algebraic geometry). 
One standard subclass, for tensors $T \in \bC^{n \times n \times p}$, is the class of \emph{$3$-generic} tensors, which is made of tensors $T \in \bC^{n \times n \times p}$ having the property that its 3-slices $T_1, \dots, T_p \in \bC^{n \times n}$ \emph{span an invertible matrix}.   
Since (border) tensor rank is invariant under certain linear changes of coordinates, we can henceforth assume without loss of generality that $3$-generic tensors $T$ are those for which the 3-slice $T_1$ is an invertible matrix.
Such classes of matrices have also been the subject of interest in the lower bounds literature, already starting with Strassen's lower bound on tensor rank via commutators \cite{strassen83}.

In order to generalize the above definition of genericity to the rectangular 3-tensor setting, that is when $T \in \bC^{m \times n \times p}$, we define the following subclass of tensors. 

\begin{definition}[Rank generic tensors]\label{definition: rank generic tensors}
    Let $T \in \bC^{m \times n \times p}$ and $r \in \bN$. 
    We say that $T$ is $(r, 3)$-generic if there exist rank $r$ matrices, $A\in \bC^{m \times r},B\in \bC^{r\times n}$ and a $3$-generic tensor $Z \in \bC^{r \times r \times p}$ (that is, $Z_1$ is invertible) such that $T_i=AZ_iB$ for all $i\in[p]$. 
    We call any such triple $(A, B, Z)$ satisfying the above conditions an $(r,3)$-generic factorization of $T$. 
\end{definition}

\begin{remark}
    Note that in the definition above we must have $r \leq \min(m,n)$, since $\Rank(Z_i) \leq \Rank(T_i) \leq \min(m,n)$.
\end{remark}

% \begin{remark}
%     Denoting by $E'(T)$ the error-degree of a tensor $T$ as in \cite[Definition 45]{dutta2025recent}, the two different definitions of \textit{error-degree} are connected in the following way. 
%     Let us start with a tensor $T\in\bC^{m\times n \times p}$ with $\Brank(T)=r$. Let $T(\e):=\frac{\sum_{i=1}^r u_i(\e)\tp v_i(\e)\tp w_i(\e)}{f(\e)}$ for some $u_i(\e),v_i(\e),w_i(\e),f(\e)\in\bC[\e]$ such that $T(\e)= T + \e Q$ for some $Q\in\bC[\e]^{m\times n \times p}$. We can factorize $f(\e)$ as $f(\e)=c\e^kg(\e)$ for some $g(\e)$ with constant term equal to 1. Define, $$T'(\e):=g(\e)T(\e) = \frac{\sum_{i=1}^r u_i(\e)\tp v_i(\e)\tp w_i(\e)}{c\e^k}$$

%     Then, $lim_{\e\rightarrow 0}T'(\e)= T$ and $\Rank(T'(\e))\leq r$ for all $\e\neq 0$. But note that $E'(T)\leq \edeg(T) +deg(g(\e))$.
% \end{remark}

We are now ready to state our main theorems, which give polynomial upper bounds on the border degeneration order as well as on the error degrees of certain $(r,3)$-generic tensors.
As we discuss in more detail later, this is an exponential improvement (for this class of tensors) on the original bounds of \cite{LEHMKUHL19891}.

\begin{restatable}{theorem}{approximateMain}\label{thm: general m not n}
Let $T\in \bC^{m\times n \times 3}$ be an $(r, 3)$-generic tensor.
If $(A, B, Z)$ is an $(r, 3)$-generic factorization of $T$ such that $Z_2Z_1^{-1}$ and $Z_3Z_1^{-1}$ commute, then $\Brank(T) = r$, $\ \edeg(T) \leq (r-1)^3 + (r-1)^2$ and $\bdorder(T) \leq 2(r-1)^3 + 3(r-1)^2 + 3(r-1)$.
% Suppose there exist rank $r$ matrices, $A\in M_{m\times r}(\bC),B\in M_{r\times n}(\bC)$ such that $T_i=AZ_iB$ for some $Z_i\in M_r(\bC)$ for all $i\in[3]$ and the tensor $Z:=[Z_1,Z_2, Z_3]$ satisfies the following conditions:
% \begin{enumerate}
%     \item $Z_1$ is invertible.
%     \item $Z_2Z_1^{-1}$ and $Z_3Z_1^{-1}$ commute
% \end{enumerate}

% \noindent Then, $\Brank(T) = r$, $\ \edeg(T) \leq (r-1)^3 + (r-1)^2$ and $\bdorder(T) \leq 2(r-1)^3 + 3(r-1)^2 + 3(r-1)$. 
% there exists $T(\e)=\sum_{i=1}^r u_i(\e)\tp v_i(\e)\tp w_i(\e)$ such that $lim_{\e\rightarrow 0}T(\e) = T$ with the entries of $T(\e)$ polynomial in $\e$ with degree at most $(r-1)^3 + (r-1)^2$ and $T$ has a degree of border degeneration at most $2(r-1)^3 + 3(r-1)^2 + 3(r-1)$.
\end{restatable}
The above result applies to tensors with 3 slices (format $m \times n \times 3$). 
Border rank for tensors with 2 slices has been well understood for a long time~\cite{bini80}. 
It appears that no systematic study of the next simplest case (3 slices) is available in the literature, a gap which we attempt to fill in the present paper.
Note that the hypothesis that $Z_2Z_1^{-1}$ and $Z_3Z_1^{-1}$ commute in Theorem~\ref{thm: general m not n} is not overly restrictive since it is {\em necessary} for the conclusion $\Brank(T) = r$ to hold true (see~\cite{strassen83} and~\cref{thm: not1reg_main}). 
A similar remark applies to the $(r,3)$-genericity hypothesis: if $\Brank(T) = r$, there must exist a factorization $T_i=AZ_iB$ where $A$ and $B$ have rank $r$ (\cref{cor: dim_shift}). 
So 3-genericity (the invertibility of $Z_1$) is the only restrictive assumption here.

Our next result works in the setting of larger tensors ($p \geq 3$), but it needs an extra assumption on top of genericity: that one of the 3-slices has somewhat well-behaved eigenspaces.
The extra concept needed here is the notion of regularity of a matrix, in particular the notion of 1-regularity: a matrix is 1-regular if every eigenspace is one dimensional.  
We are able to show that $(r,3)$-generic tensors in $\bC^{m \times n \times p}$ which satisfy this extra assumption also have border rank $r$, and in this case we are able to show that such tensors have error degree 1, which in particular implies that such tensors have rank upper bounded by $2r$.
For more on regularity, we refer the reader to \cref{subsection: regular matrices prelim}.

\begin{restatable}{theorem}{linearErrorMain}\label{thm: linear part m not n}
Let $T\in \bC^{m\times n \times p}$ be an $(r, 3)$-generic tensor.
If $(A, B, Z)$ is an $(r, 3)$-generic factorization of $T$ such that  
% Suppose there exist rank $r$ matrices, $A\in M_{m\times r}(\bC),B\in M_{r\times n}(\bC)$ such that $T_i=AZ_iB$ for some $Z_i\in M_r(\bC)$ for all $i\in[p]$ and the tensor $Z:=[Z_1,Z_2,\dots, Z_p]$ satisfies the following conditions:
\begin{enumerate}
    % \item $Z_1$ is invertible.
    \item $Z_2Z_1^{-1}$ is 1-regular.
    \item $Z_iZ_1^{-1}$ and $Z_jZ_1^{-1}$ commute for all $i,j\in [p]$
\end{enumerate}

\noindent Then, $\Brank(T) = r$ and $\edeg(T) \leq 1$. 
Moreover, $\bdorder(T) \leq r-1$.
% Then, there exists $T(\e)$ with entries linear in $\e$ and $\Rank T(\e) \leq r$ for all but finitely many $\e$
% $=\sum_{i=1}^r u_i(\e)\tp v_i(\e)\tp w_i(\e)$ 
% such that $\lim_{\e\rightarrow 0}T(\e) = T$. Furthermore, $\bdorder(T)\leq r-1$.
% with the entries of $T(\e)$ linear in $\e$.
\end{restatable}
Like in~\cref{thm: general m not n}, the assumption that the matrices $Z_iZ_1^{-1}$ commute is not really restrictive: it is again necessary for the conclusion $\Brank(T)=r$ to hold true 
(see \cref{thm:border1reg}).

\paragraph{Comparison with previous work \& debordering.}
The upper bounds of $2(r-1)^3 + 3(r-1)^2 + 3(r-1)$ in \cref{thm: general m not n} and of $r-1$ in \cref{thm: linear part m not n} that we obtain for the \textit{degeneration order} of $T$ improves on the general exponential upper bound for 3-tensors by \cite{LEHMKUHL19891}. 
For a tensor $T$ of shape $m\times n \times p$ and border rank $r$, they give the upper bound of $3^{(m+n+p-3)r}$ for its \textit{degeneration order}. 
While these bounds on degeneration order still fail to provide a meaningful debordering result for 3-tensors when $r$ is large enough (for instance for concise generic tensors), note that our bound on the \emph{error degree} in \cref{thm: linear part m not n}, when combined with standard interpolation results (\cref{lem: interpolation}), shows that tensors satisfying the conditions of \cref{thm: linear part m not n} have tensor rank upper bounded by~$2r$, therefore yielding non-trivial debordering results for 3-tensors.

The results in~\cite{LEHMKUHL19891} rely on an upper bound on the degree of the variety of tensors of border rank at most $r$. We rely instead on more specific properties of this algebraic variety, and in particular on commutativity properties of certain matrices derived from the tensor slices (as explained in~\cref{sec:MT} and~\cref{sec:overview}).

\paragraph{Comparison with debordering in the symmetric tensor setting.}
The recent works \cite{dutta2024fixed,shpilka2025improved} have obtained improved debordering results in the setting of symmetric tensors -- i.e. (border) Waring rank of polynomials.
In their setting, given a symmetric tensor (i.e. a homogeneous polynomial) $T \in (\bC^{n})^{\otimes d}$ with border symmetric rank $r$, \cite{dutta2024fixed} shows that the symmetric rank of the tensor is upper bounded by $4^r \cdot d$, and this bound was improved by \cite{shpilka2025improved} to $d \cdot r^{10 \sqrt{r}}$.
%In our setting, we show 
Our results imply that any symmetric $3$-tensor (i.e. $d = 3$) satisfying the conditions of \cref{thm: linear part m not n} has symmetric tensor rank upper bounded by $O(r)$. This follows from
the aforementioned $2r$ tensor rank upper bound, together with the fact that the symmetric tensor rank of any order-$d$ symmetric tensor $T$ is bounded by $2^{d-1} \Rank(T)$: 
see e.g.~\cite[Section 2.1]{koiran2020tensor}.

\paragraph{Overcomplete setting.}
While the above \cref{thm: general m not n,thm: linear part m not n} work only
%\PK{is there a good reason for the word ``mostly'' here?} 
in the \emph{undercomplete setting} (i.e. when the rank of the tensor is upper bounded by the dimensions of the tensor), we are also able to extend some of our debordering results to the \emph{overcomplete setting} (i.e. when the rank is allowed to be larger than the ambient dimensions).
This extension is given by \cref{th:over1regular,th:overnot1regular}.
It relies on a characterization of border rank via commuting extensions
which generalizes a result in \cite{koi24overcomplete}.
We can then apply our results from the undercomplete setting 
to the commuting extension in order to bound the error degree 
and order of degeneration in the overcomplete setting.

%and is achieved by combining the tools used in the undercomplete setting along with: 
%\begin{itemize}
%    \item a characterization of border rank via commuting extensions, generalizing a result in \cite{koi24overcomplete}
%    \item the 3-slices belonging to certain projections of the closure of diagonalizable matrices. 
%\end{itemize}
As these results are fairly technical, we refer the interested reader to \cref{sec:over}, where we formally state and prove these results.
Quantitatively, the bounds that we obtain are of the same shape as for undercomplete decomposition: error degree 1 and order of degeneration at most $r-1$ in the 1-regular case; error degree and order of degeneration $O(r^3)$ without this assumption.

%================================================================================
\subsection{The Motzkin-Taussky Theorem} \label{sec:MT}
%================================================================================
The Motzkin-Taussky Theorem \cite[Theorem 5]{MT1955} (stated in \cref{Thm: Motzkin-Taussky}) plays a crucial role in the proof of Theorem~\ref{thm: general m not n}, and the connection will be explained in~\cref{sec:overview}.

To present the Motzkin-Taussky theorem, we need the concept of approximately simultaneously diagonalizable (ASD) matrices. Informally speaking, a pair of matrices is said to be ASD if it is possible 
to turn them into simultaneously diagonalizable matrices by perturbing
them infinitesimally (see \cref{sec:asd} for a more formal presentation).  The Motzkin-Taussky Theorem states that a pair of matrices is ASD if and only if these matrices commute.
Recall that simultaneously diagonalizable matrices always commute.
From this, it follows easily by continuity that any pair of ASD matrices must commute. The converse is the nontrivial part of the theorem. 

For~\cref{thm: general m not n} we need a constructive 
version of the Motzkin-Taussky Theorem. We are aware of 3 different proofs of this theorem: the original proof in~\cite{MT1955} is by induction on the dimension. A second inductive proof was proposed in~\cite{OV06}. Finally, there is a direct (non inductive) proof by Guralnick~\cite{Guralnick92}. These 3 proofs are presented in~\cite{OMV11}. We take our inspiration from Guralnick's proof because the inductive proofs seem to give worse quantitative results.
Note that the result actually proved in~\cite{Guralnick92} is that the variety of pairs of commuting matrices is irreducible. This is another
theorem from the Motzkin-Taussky paper~\cite[Theorem~6]{MT1955},
which is known~\cite[Theorem 7.5.2]{OMV11} to be equivalent to the Motzkin-Taussky theorem as stated above 
(i.e., to \cite[Theorem 5]{MT1955}).
In the proof of~\cref{thm: general m not n} we build on Guralnick's argument to construct explicit perturbations with good quantitative properties.

%================================================================================
\subsection{Proof Overview} \label{sec:overview}
%================================================================================

In this subsection we give the reader an overview of the main concepts and tools used to derive our main results in the undercomplete setting (\cref{thm: general m not n,thm: linear part m not n}), along with a description of our main technical results.

Our starting point comes from a combination of Strassen's lower bound on border rank via commutators \cite{strassen83}\footnote{Strassen's original lower bound was only for tensor rank, but one can see that it also lower bounds the border rank.} and the Motzkin-Taussky theorem on approximately simultaneously diagonalizable matrices \cite[Theorem 5]{MT1955} %(stated in \cref{Thm: Motzkin-Taussky}). 
(see~\cref{sec:MT}).
Combining these two results, one shows that if $T \in \bC^{n \times n \times 3}$ is a tensor with the first 3-slice being identity and such that all 3-slices commute, then $\Brank(T) = n$.
In this case, where we know the border rank of our tensor, one can ask: what are its error degree and its degeneration order?
\footnote{Unfortunately, the known proofs of Motzkin-Taussky do not give explicit bounds on the error degree nor on the degeneration order.}

In the above case, the approximation questions become equivalent to finding \say{low complexity} approximately simultaneously diagonalizable perturbations of our matrices. 
So let us take our tensor to have 3-slices $I, A, B \in \bC^{n \times n}$ where $AB = BA$.
If $A$ is a single Jordan block, a direct computation of the commutation relation shows that $B$ must be an upper triangular Toeplitz matrix, and thus, $B \in \bC[A]$ (i.e. the $\bC$-algebra generated by $A$).
This is equivalent to saying that there is a polynomial $p(z) \in \bC[z]$ such that $B = p(A)$.
Hence, all we need to do in this case is to find a low complexity perturbation $A(\e)$ of $A$ which is diagonalizable, and this would give us a perturbation $B(\e) := p(B(\e))$ of $B$ which is simultanously diagonalizable with $A(\e)$.

The most direct way of obtaining such diagonalizable approximation of $A$ is to simply perturb the diagonal entries of $A$ so that the resulting matrix $A(\e)$ has distinct diagonal entries. 
This would force $A(\e)$ to have distinct eigenvalues (which are the diagonal entries of $A(\e)$, but it would only give us a bound of $\edeg(B) \leq n$, which is too expensive for the debordering question.
Could there be a better perturbation that still yields $A(\e)$ to have distinct eigenvalues but lowers the complexity of $B(\e)$?
This is what we do in \cref{subsection: special perturbation} (in a bit more generality): note that the perturbation $A + \e \cdot E_{n,1}$, where $E_{n,1}$ is the elementary matrix with a 1 in the $(n,1)$ entry and zeros everywhere, also has $n$ distinct eigenvalues.
Moreover, as $A$ is a Jordan block, one can show that $A(\e)^k$ is still linear in $\e$ for all powers of $k$, which implies that $B(\e)$ will also be linear in $\e$!
Thus, this perturbation shows that $\edeg(T) \leq 1$ in our case, and we obtain a much improved debordering result (i.e. $\Rank T \leq 2n$) as a consequence.
We can also obtain a good bound on $\bdorder(T)$ (in particular $\leq n-1$) by simply noticing that the following perturbation $A(\e) := A + \e^n E_{n,1}$ has polynomial eigenvalues in $\e$! 
The above argument is the essence of \cref{prop: linear_part}.

But the challenge in our setting is that it may be the case that neither $A$ nor $B$ are a single Jordan block (thus having a much more complex Jordan structure).
This is what makes \cite[Theorem 5]{MT1955} very interesting, and is what also presents us with a significant challenge. 
Our approach to overcome this challenge goes by answering two questions: can we generalize the case where $A$ is a single Jordan block and still obtain and approximation of $B$ which is a polynomial function of the corresponding approximation of $A$?
When neither $A$ and $B$ satisfy the condition in this last question, can we find a matrix $C$ which approximates $A$ and satisfies the generalized condition and commutes with $B$?
We are able to answer both questions in the affirmative, by combining the concept of regular matrices (in particular 1-regular matrices - see \cref{subsection: regular matrices prelim}) along with the Weyr form -- another normal form for matrices which is related to the Jordan normal form, but behaves better with regards to commutativity (see \cref{subsection: weyr form}).\footnote{The concepts of regular matrices and the Weyr form have been well-studied from the perspective of commuting matrices and approximate simultaneous diagonalization, and it is thus natural that they will also be useful in the study of tensor (border) rank.}

The concept of regular matrices quantifies how large all eigenspaces of a particular matrix are: a matrix is $k$-regular if every eigenspace is at most $k$-dimensional.
Thus, for instance, a 1-regular matrix is simply any matrix which is similar to a direct sum of simple (and non-trivial) Jordan blocks. 
This fact, along with other properties of $1$-regular matrices, allows us to directly generalize our approach above of linear perturbations to the setting where one of the 3-slices of our tensor is 1-regular, and this is enough for us (with a bit more work) to generalize the above approach to prove \cref{thm: linear part m not n}.
For the general case, a key lemma that we need is \cref{lem:1-reg ext}, which tells us that for any matrix $A$, there is a $1$-regular matrix $R$ which commutes with $A$.
With this fact at hand, we can now use the Weyr form to simultaneously upper triangularize $A, B, R$, and with some more work we can show that an appropriate diagonal perturbation of $B + \e \cdot R$ (which will no longer be linear in $\e$) will approximately simultaneously diagonalize both $A$ and $B$, which yields the proof of \cref{thm: general m not n}.
Because of the above detour via the 1-regular matrix $R$, our bounds are a bit worse in this case, as more care needs to be taken at every step along the way.

%================================================================================
\subsection{Further work}
%================================================================================

The results in this paper raise a number of interesting questions. 
First and foremost, to what extent can our results on tensors with 3 slices (format $m \times n \times 3$) be extended to tensors with more slices? One roadblock in this direction is that there is no known analogue of the Motzkin-Taussky theorem for more than 2 matrices. 
Already for 3 matrices, it is known that commutativity does not imply the ASD property. 
Thereby, even for generic tensors in $\bC^{n \times n \times 4}$, it is known that the commutativity relations alone do not necessarily yield minimal border rank. 
Some further necessary conditions beyond commutativity were discovered, but no sufficient and necessary condition has been found yet. 
On all of this, the book~\cite{OMV11} is highly recommended.

Even for tensors with 3 slices, our results are not quite complete. 
In the overcomplete setting, further progress may hinge on a better understanding of the geometry of tuples of commuting matrices, and in particular on the question whether certain projections are closed:  see~\cref{sec:over}, and in particular~\cref{rem:closed} and~\cref{cor:projclosed}. 
In the undercomplete setting, our results are fairly general except for the fact that we need to assume that certain matrices are invertible. 
Moreover, even under the assumptions of~\cref{thm: general m not n} it might well be possible to find better upper bounds on the degree of error and on the order of degeneration. 
One important thing to notice here is that via padding arguments, if one removes the genericity assumption, then one is solving the general debordering question for tensor rank.
Thus, it would also be interesting to relax the notion of genericity and obtain better debordering results for such \say{less generic} examples.

The above leads to the open question of whether Koszul-Young flattenings can be used to obtain upper bounds on the degree of error and on the order of degeneration, as they are an alternative lower bound method which does not require the genericity assumption. 
Like commuting extensions,  Koszul-Young flattenings have been used to obtain lower bounds on tensor rank~\cite{landsberg15,landsbergGCT} and decomposition algorithms~\cite{kothari2025overcomplete}. 
One advantage of flattenings is that they do not require taking a matrix inverse, which is a limitation of our current methods.

In another direction, one could try to generalize \cref{thm: linear part m not n}. Recall that this theorem relies on the assumption that a certain matrix is 1-regular, which allows us to analyze tensors with more than 3 slices. Can we say something if this assumption is relaxed (for instance, if it is replaced by a 2-regularity assumption)? 

Finally, we note that in all of our results the upper bound on the order of degeneration is  larger than the upper bound on the degree of error. Is this a general property, or is it just an artifact of our proof methods? We do know that the error degree is always at most twice the degeneration order (\cref{lem:errorfromdegeneration}).

\subsection{Organization of the paper}

In \cref{sec:prelim}, we setup the notation used throughout the paper and also establish important facts used throughout the paper.
In \cref{sec:asd} we define approximate simultaneous diagonalisability, present the Motzkin-Taussky theorem, and establish useful facts about our special perturbation.
In \cref{sec:border} we establish our main technical results and prove the main theorems in this paper.

%================================================================================
\section{Preliminaries} \label{sec:prelim}
%================================================================================

In this section we collect various facts on matrices and tensors. We recommend to consult it when needed rather than read it in a linear fashion.

We denote by $\bC$ the field of complex numbers, and for a ring $R$, we denote by $M_n(R)$ the set of $n \times n$ matrices over $R$.
For non-square matrices, $M_{m,n}(R)$ denotes the set of $m \times n$ matrices over $R$.
Below are some conventions and definitions that we adopt in this paper.

\begin{itemize}
    \item For $i \in [m], j \in [n]$, denote by $E_{i,j}$ to be the elementary matrix with $(i,j)$ entry equal $1$ and all other entries being zero.
    \item Given $A \in M_n(\bF)$, let $vec:M_n(\bF)\longrightarrow \bF^{n^2}$ be the function that takes a matrix $M=(m_{i,j})$ to $vec(M)=(v_i)_{i\in [n^2]}$ where $v_{i+n(j-1)}:=m_{i,j}$ for all $i,j\in [n]$.
    
    \item Given $(\lambda_1, \dots, \lambda_n) \in R^n$, denote by $\diag(\lambda_1, \dots, \lambda_n)$ the diagonal matrix with $(i,i)$ entry equal $\lambda_i$.
    That is, $\diag(\lambda_1, \dots, \lambda_n) = \sum_{i=1}^n \lambda_i E_{i,i}$.
    
    \item For any $S\subseteq [n]$, define the map $\pi_S:\bC^{n}\longrightarrow \bC^{\vert S\vert}$ to be the projection map onto the coordinates corresponding to the elements of $S$.
    
    \item Given $A \in M_n(\bC)$, denote by $\mathcal{C}(A)$ the \emph{centralizer} of $A$ in $M_n(\bC)$, that is $\cC(A)$ is the subset of matrices in $M_n(\bC)$ which commute with $A$.

    \item A matrix $A \in M_n(\bC)$ is called \emph{Toeplitz} if for every $i,j\in [n-1]$, $A_{i,j} = A_{i+1,j+1}$.

    \item For a matrix $M\in M_n(\bF)$ diagonalizable over $\bF$, we will call $U\in M_n(\bF)$ a \emph{matrix of eigenvectors of $M$} if the columns of $U$ form an eigenbasis of $M$.
    
    \item We will denote by $W_n$ to be the single Jordan block of size $n$ (and we will write $W$ for short when $n$ can be understood from context).
    That is,
    \begin{linenomath}
        $$
    W_n :=
    \sum_{i=1}^{n-1} E_{i, i+1} =
    \begin{bmatrix}
        0 & 1 & \cdots & \cdots & 0\\
        0 & 0 & 1 & \cdots & 0\\
        \vdots & & \ddots & \ddots & \vdots\\
        0 & & & \ddots & 1\\
        0 & \cdots & \cdots & \cdots& 0
    \end{bmatrix}.
    $$
    \end{linenomath}
    
\end{itemize}
\medskip

%================================================================================
\subsection{The Weyr form of a matrix}\label{subsection: weyr form}
%================================================================================

The Weyr form of a matrix is an alternative normal form to the Jordan form that has several interesting properties, especially with regards to the study of commuting matrices. 
We will use the Weyr form to simultaneously triangularize commuting matrices.
For a more comprehensive survey and study of the Weyr form, we refer the reader to \cite{OMV11}, which is where we took most of the definitions and results in this subsection.

\begin{definition}
    A {\bf basic Weyr matrix with eigenvalue $\lb$} of size $n\times n$ is a matrix $M$ that has a block structure with blocks $M_{ij}$ of size $n_i\times n_j$ for $1\leq i,j \leq r$ for some $r$-partition $n = n_1 + n_2 + \dots n_r$ of $n$ with $n_1\geq n_2\geq \dots \geq n_r$ that satisfies the following properties:

    \begin{enumerate}
        \item The blocks $M_{ii}$ of size $n_i\times n_i$ are the matrices $\lb I_{n_i}$.
        \item The blocks $M_{i,i+1}$ of size $n_i\times n_{i+1}$ are full column rank matrices in reduced row-echelon form.
        \item All other blocks are zero matrices.
    \end{enumerate}

    \noindent In this case, we say that $M$ has the Weyr structure $(n_1,n_2,\dots, n_r)$.
\end{definition}

\noindent The following are examples of basic Weyr matrices.

\begin{equation*}
 \begin{bNiceArray}{cc|[end=4]cc|c}[margin]
    \lb & 0 & 1 & 0 & \\
    0 & \lb & 0 & 1 & \\
    \Hline
    && \lb & 0 & 1 \\
    && 0 & \lb & 0 \\
    \Hline[start=3]
    &&&& \lb
\end{bNiceArray}\tag{Basic Weyr matrix with structure $(2,2,1)$}   
\end{equation*}
\begin{equation*}
 \begin{bNiceArray}{ccc|[end=5]cc|cc}[margin]
    \lb & 0   & 0   & 1   & 0   &    &\\
    0   & \lb & 0   & 0   & 1   &    &\\
    0   & 0   & \lb & 0   & 0   &    &\\
    \Hline
        &     &     & \lb & 0   & 1 & 0\\
        &     &     & 0   & \lb & 0 & 1\\
    \Hline[start=4]
    &&&&&\lb & 0\\
    &&&&& 0 & \lb
\end{bNiceArray}\tag{Basic Weyr matrix with structure $(3,2,2)$}   
\end{equation*}

\begin{definition}
    A matrix $M$ is said to be in \emph{Weyr form} (or $M$ is a \emph{Weyr matrix}) if, 
    \begin{equation*}
        M =  \begin{bNiceArray}{c|[end=2]c|[start=2, end=2]c|[start=4]c}[margin]
                M_1 & & &\\
                \Hline[end=2]
                & M_2 & &\\
                \Hline[start=2, end=2]
                & & \ddots &\\
                \Hline[start=4]
                & & & M_k\\
            \end{bNiceArray}
    \end{equation*}
    where $M_i$ are basic weyr matrices with distinct eigenvalues.
\end{definition}

\noindent The next lemma appears in \cite[Section 2.1, p. 48]{OMV11}.

\begin{lemma}\label{lem: J_to_W}
    Let $J$ be a Jordan matrix. Then there exists a permutation matrix $P$ such that $M:=P^{-1}JP$ is a Weyr matrix.
\end{lemma}

The next theorem is taken from \cite[Theorem 2.2.4]{OMV11}, which follows from the above lemma together with the theorem on the Jordan normal form.

\begin{theorem}
    Up to permutation of the basic Weyr blocks, each square matrix $A$ over an algebraically closed field is similar to a unique Weyr matrix $\widetilde{A}$ which is called the Weyr form of $A$.
\end{theorem}

The next theorem appears in \cite[Theorem 2.3.5]{OMV11}.
This result is an example of what makes the Weyr form a better normal form when studying commuting matrices.

\begin{theorem}\label{thm: weyr_upper_triangular}
    Let $A_1,A_2,\dots, A_k$ be commuting $n\times n$ matrices over an algebraically closed field. Then there is a $n\times n$ matrix $C$ such that $C^{-1}A_1 C$ is a Weyr matrix and the matrices $C^{-1}A_i C$ are upper triangular for all $2\leq i \leq k$
\end{theorem}

The above theorem allows us to put one of the commuting matrices in Weyr form and the other matrices in upper triangular form. 
The Jordan normal form does {\em not} have this property~\cite[Examples 2.3.6 and Remark 2.3.7]{OMV11}, and this is the reason why we work with the Weyr form.

%================================================================================
\subsection{Regular Matrices}\label{subsection: regular matrices prelim}
%================================================================================

\begin{definition}
    A matrix $A\in M_n(\bF)$ is said to be $k$-regular if every eigenspace of the matrix $A$ as a matrix in $M_n(\ol{\bF})$ has dimension at most $k$.
\end{definition}

1-regular matrices have very nice algebraic properties, some of which are captured by the following proposition, which can be found in \cite[Proposition 1.1.2]{OMV11}.

\begin{proposition}\label{prop:centralizer}
    The following are equivalent:
    \begin{enumerate}
        \item $A \in M_n(\bC)$ is a 1-regular matrix,
        \item $I, A, A^2,\dots, A^{n-1}$ are linearly independent, 
        \item $\mathcal{C}(A)= \bC[A]$.
    \end{enumerate} 
\end{proposition}

\begin{lemma}\label{lem:1-reg ext p2}
    For any 1-regular matrix $R\in M_n(\bC)$ and any matrix $B\in M_n(\bC)$, $R+\alpha B$ is 1-regular for all but finitely many $\alpha \in \bC$.
\end{lemma}

\begin{proof}
    Set $A(\alpha) := R +\alpha B$. 
    By \cref{prop:centralizer}, $A(\alpha)$ is 1-regular iff $I, A(\alpha), \dots, A(\alpha)^{n-1}$ are linearly independent. 
    Let $M(\alpha)$ be the following $n^2 \times n$ matrix. 
        $$M(\alpha) =
        \begin{pmatrix}
            vec(I) & vec(A(\alpha)) & vec(A(\alpha)^2) & \dots & vec(A(\alpha)^{n-1})
        \end{pmatrix}.$$
    Note that $I, A(\alpha), \dots, A(\alpha)^{n-1}$ are linearly independent iff $M(\alpha)$ is of full rank.
    Since $R = A(0)$ is 1-regular, $M(0)$ is of full rank, which implies that there is a subset $S \subset [n^2]$ such that the minor $\det M(0)_{S, [n]} \neq 0$.
    Thus, we have $\det M(\alpha)_{S, [n]}$ is a nonzero polynomial in $\alpha.$
    In particular, this implies that $\det M(\alpha)_{S, [n]}$ has finitely many roots in $\bC$. 
    Since $M(\alpha)$ is of full rank whenever $\det M(\alpha)_{S, [n]} \neq 0$, we conlcude that $A(\alpha)$ is 1-regular for all but finitely many $\alpha\in\bC$.
\end{proof}

The following technical results are used in proving \cref{thm: not_1-regular}.

\begin{proposition}\label{prop: poly_coeff}
    Let $A(\e),B(\e)\in M_n(\bC[\e])$ be such that $A(\e) \in \bC(\e)[B(\e)]$. 
    Then, there exists $p(x)=\sum_{i=0}^{n-1}a_i(\e)x^i\in\bC(\e)[x]$ with $a_i(\e)=\frac{f_i(\e)}{g(\e)}$ for all $0 \leq i \leq n-1$ with $deg_\e(f)\leq (n-1)^2d_B^2 + d_A$ and $deg_\e (g(\e))\leq n(n-1)d_B$ where $d_A:=max_{i,j}(deg_\e(A_{i,j}(\e)))$ and $d_B:=max_{i,j}(deg_\e(B_{i,j}(\e)))$.
\end{proposition}

\begin{proof}
    Let $s$ be the degree of the minimal polynomial of $B(\e)$.
    \noindent We have,
    \begin{align*}
     A(\e)=p(B(\e)) &= \sum_{i=0}^{s-1} a_i(\e)(B(\e)))^i
    \end{align*}

    \noindent For $i\in[n]$, define $v_i:= vec((B(\e))^{i-1})$. Also, let $v_A:=vec(A(\e))$. 
    Then, we have the following.
    
    \begin{equation}\label{eq: coeff_relation-1}
     v_A = \begin{bmatrix}
          v_1 & v_2 & \dots & v_s
     \end{bmatrix} \begin{bmatrix}
         a_0(\e)\\
         a_1(\e)\\
         \dots \\
         a_{s-1}(\e)
     \end{bmatrix}\\
    \end{equation}
    
    We know that a solution to system $(\ref{eq: coeff_relation-1})$ has to be unique because the degree of the minimal polynomial of $B(\e)$ over $\bC(\e)$ is $s$. Therefore, there exists a subset $S\subset [n^2]$ of size $s=rank ([v_1\;\; v_2\;\; \dots \;\; v_n])$ for which the following is true.
    
     \begin{equation}
         \pi_S(v_A) = \begin{bmatrix}
              \pi_S(v_1) & \pi_S(v_2) & \dots & \pi_S(v_s)
         \end{bmatrix} \begin{bmatrix}
             a_0(\e)\\
             a_1(\e)\\
             \dots \\
             a_{s-1}(\e)
         \end{bmatrix} \\
     \end{equation}
    
    \noindent and $\begin{bmatrix}
              \pi_S(v_1) & \pi_S(v_2) & \dots & \pi_S(v_s)
         \end{bmatrix}$ is invertible, where $\pi_S$ is the projection map onto the set of coordinates $S$. 
    
    Let $M=\begin{bmatrix}
              \pi_S(v_1) & \pi_S(v_2) & \dots & \pi_S(v_s)
         \end{bmatrix}$ 
    and $a(\e)=(a_0(\e),\dots,a_{s-1}(\e))$. Since $M$ is invertible, $a(\e)= M^{-1}\pi_S(v_A)$.
    Let $C$ be the matrix of cofactors of $M$. Then, $M^{-1}= C^t/det(M)$.\medskip
    
    Since $s\leq n$, we get the following bounds. The entries of $M$ are polynomial in $\e$ with degree at most $(n-1)d_B$. Therefore, the entries of $C^t$ are polynomial in $\e$ with degree at most $(n-1)^2d_B^2$. 
    Furthermore, $det(M)$ is a polynomial in $\e$ with degree at most $n(n-1)d_B$. 
    It follows that the coefficients $a_i(\e)$ are rational in $\e$ of the form $f(\e)/det(M)$ where $f(\e)\in \bC[\e]$ with $deg_\e(f)\leq (n-1)^2d_B^2 + d_A$ and $\deg_\e (det(M))\leq n(n-1)d_B$.
\end{proof}

\begin{lemma}\label{lem: UT_diag}
    Let $B,R\in M_n(\bC)$ be upper triangular matrices such that $R$ is 1-regular. Then, for any  non-zero $f(\e)\in\bC[\e]$, there exists a diagonal matrix $D\in M_n(\bC)$ such that $B(\e):= B+\e R +\e f(\e)D$ is diagonalizable for all but finitely many $\e$.
\end{lemma}

\begin{proof}
    $B+\e R +\e f(\e)D$ is upper triangular and hence is diagonalizable exactly when it has distinct diagonal entries. If $f(\e)$ is a non constant polynomial in $\e$, then any choice of $D$ with distinct diagonal entries makes $B(\e)$ diagonalizable and when $f(\e)$ is a constant polynomial, $D$ can be chosen such that $R+f(\e)D$ has distinct diagonal entries. 
    
    Observe that $B(\e)$ is diagonalizable for all but finitely many $\e$ because $B(\e)$ has distinct diagonal entries for all but finitely many $\e$ and thus we have proved the lemma. \footnote{One can show that $B(\e)$ remains diagonalizable for all but finitely many $\e$ even if $B$ and $R$ are not assumed to be upper triangular. This is due to our choice of $D$ and to the fact that the set of (diagonalizable) matrices with simple eigenvalues only is Zariski open.  See Remark~\ref{rem:weyr} for further comments.}
\end{proof}

%================================================================================
\subsection{Valuations and nice matrices} \label{sec:valuation}
%================================================================================

For the field $\bC(x)$, we define the function $\nu_x:\bC(x)\rightarrow \bZ\cup\{\infty\}$ as follows: for all non-zero polynomials $f(x),g(x)\in\bC[x]$, let $\nu_x(f(x)) := r$ where $r$ is the minimum of the degrees of monomials in $f(x)$ with non-zero coefficient and $\nu_x(f(x)/g(x)) := \nu_x(f(x))-\nu_x(g(x))$.
It is a standard fact that $\nu_x$ is a valuation, and thus satisfies the following properties:
\begin{enumerate}
    \item $\nu_x(a)=\infty$ iff $a=0$.
    \item $\nu_x(ab)=\nu_x(a)+\nu_x(b)$.
    \item $\nu_x(a+b)\geq min(\nu_x(a),\nu_x(b))$ with equality holding when when $\nu_x(a)\neq \nu_x(b)$.
\end{enumerate}

\begin{definition}\label{defn: nice-matrix}
    An upper triangular matrix $M=(m_{ij})_{i,j\in[n]}\in GL_n(\bC(x))$ is said to be $k$-\textit{nice} for some $k\in \bN$ if it satisfies the following properties:
    \begin{enumerate}
        \item $M$ is unitriangular (Upper triangular with '$1$'s in the main diagonal).
        \item $\nu_x(m_{i,i+l})\geq -kl$ for all $0\leq l\leq n-1$ and $i\in[n-l]$.
    \end{enumerate}

    $M$ is said to be $k$-\textit{semi-nice} if it satisfies just the second property listed above.
\end{definition}

\begin{lemma}
    Let $M$ be a $k$-nice matrix in $GL_n(\bC(\e))$. Then $M^{-1}$ is also $k$-nice.
\end{lemma}

\begin{proof}
    First, we show that if $A$ and $B$ are $k$-\textit{semi-nice} matrices then $AB$ and $A+B$ are also $k$-\textit{semi-nice}.
    
    Suppose $A=(a_{ij})_{i,j\in [n]}$ and $B=(b_{ij})_{i,j\in [n]}$. Let $C:=AB$ so that $c_{ij}:=\sum_{s=1}^n a_{is}b_{sj}$. When $i<j$, $c_{ij}=\sum_{s=i}^j a_{is}b_{sj}$. Fixing $i,j,s$, $\nu_x(a_{is})\geq -k(s-i)$ and $\nu_x(b_{sj})\geq -k(j-s)$. Then, $\nu_x(a_{is}b_{sj}) \geq -(k(s-i)+k(j-s))\geq -k(j-i)$ which doesn't depend on $s$. Therefore, $\nu_x(c_{ij})=\nu_x(\sum_{s=i}^j a_{is}b_{sj})\geq -k(j-i)$. The second property of $k$-\textit{nice} matrices is satisfied, and consequently $C$ is indeed a $k$-\textit{semi-nice}.

    Let $D:=A+B$ so that $d_{ij}:= a_{ij}+b_{ij}$. Then, $\nu_x(d_{ij})\geq min(\nu_x(a_{ij}),\nu_x(b_{ij}))\geq -k(j-i)$. Therefore, the sum of two $k$-\textit{semi-nice} matrices is $k$-\textit{semi-nice}.
    
    Let us now show that $M^{-1}$ is $k$-nice.  We can write $M=I+N$ where $N$ is nilpotent. More specifically, $N^n=0$. Therefore, $(I+N)(I+ \sum_{r=1}^{n-1}(-1)^{r}N^r)= I + (-1)^{n-1}N^n=I$ and so $M^{-1}=I+ \sum_{r=1}^{n-1}(-1)^{r}N^r$. $N^r$ is $k$-\textit{semi-nice} for all $r\in[n-1]$ and is upper triangular with diagonal entries 0. Therefore $M^{-1}$ is a $k$-\textit{nice} matrix since it is $k$-\textit{semi-nice} and is unitriangular.
\end{proof}

\begin{proposition}\label{prop: k-nice_ev}
    Let $B(\e)\in \bC[\e]^{n\times n}$ be an upper triangular matrix with distinct diagonal entries $\lb_i(\e)\in\bC[\e]$ with $deg(\lb_i(\e))\leq k$ for $i\in[n]$. Then, the matrix $U(\e)=[u_1(\e),\dots u_n(\e)]$ of eigenvectors of $B(\e)$ is $k$-nice.
\end{proposition}
\begin{proof}
    
   Since $B(\e)$ has distinct eigenvalues and is diagonalizable, it has $n$ distinct eigenvectors. Let $u_i(\e)$ be an eigenvector of $B(\e)$ with eigenvalue $\lb_i(\e)$ for $i\in [n]$. Then, $(B(\e)- \lb_i(\e)I_n)u_i = 0$ and $B(\e)-\lb_i(\e)I_n$ is upper triangular with distinct diagonal entries and $i^{th}$ diagonal entry being zero. Define $[\beta_{c,d}]_{c,d\in[n]}:=B(\e)-\lb_i(\e) I_n$. 
   
   From back substitution, $(u_i(\e))_l=-\frac{\sum_{j=l+1}^{n}\beta_{l,j}(u_{i}(\e))_j}{\beta_{l,l}}$ for all $l\in [i-1]$, $(u_i(\e))_l = 0$ for $l>i$ and $(u_i(\e))_i=1$. Here, note that $\beta_{l,l}\neq 0$ for all $l\neq i$ because $B(\e)$ has distinct diagonal entries. Therefore, if $u_i(\e)$ is an eigenvector of eigenvalue $\lb_i(\e)$, then $\nu_\e((u_i(\e))_j)\geq -k(i-j)$ for $1\leq j \leq i$.

   Consequently, the matrix $U(\e)$ is unitriangular and by the bounds on the valuation above, $U(\e)$ is $k$-\textit{nice} as per \cref{defn: nice-matrix}.
\end{proof}

%================================================================================
\subsection{Facts on rank and border rank of tensors}
%================================================================================

In this subsection we collect various facts about the rank and border rank of tensors.
We begin with the following fact, which appears in \cite[Prop 14.45]{BCS97}.

\begin{proposition}\label{th:slices}
    For a tensor $T=[T_1,\dots,T_p]\in\bC^{m\times n \times p}$, $\Rank(T)\leq r$ iff there exists matrices $A\in\bC^{m\times r},B\in\bC^{r\times n}$ and diagonal matrices $D_1,\dots, D_p \in \bC^{r\times r}$ such that $T_k=AD_kB$ for all $k\in[p]$.
    Moreover, $T=\sum_{i=1}^r u_i \tp v_i \tp w_i$ if and only if 
    $T_k=UD_kV^T$ for all $k\in[p]$. Here, $U$ is the matrix with 
    $u_1,\ldots,u_r$ as column vectors, $V$ is the matrix with 
    $v_1,\ldots,v_r$ as column vectors, 
    and $D_k = \diag(w_{1k},\ldots,w_{nk})$.
\end{proposition}

\begin{corollary}\label{cor: rank_n_char}
    For a tensor $T=[T_1,\dots,T_p]\in\bC^{n\times n \times p}$ with $T_1$ invertible, $\Rank(T)= n$ iff the matrices $T_2T_1^{-1},\dots, T_pT_1^{-1}$ are simultaneously diagonalizable. 
    Moreover, if we have the simultaneous diagonalization
    \begin{equation}     \label{eq:simdiagcor}
    T_kT_1^{-1}=UD_kU^{-1},\ k=2,\ldots,p
    \end{equation}
    then $T$ can be decomposed as $T=\sum_{i=1}^n u_i \tp v_i \tp w_i$
    where $u_1,\ldots,u_n$ are the columns of $U$ and $v_1,\ldots,v_n$ are the columns of $V=(U^{-1}T_1)^T$. The vectors $w_1,\ldots,w_n$
    are defined by: $w_{i1}=1$ for $i=1,\ldots,n$ 
    and $D_k = \diag(w_{1k},\ldots,w_{nk})$ for $k=2,\ldots,p$.
\end{corollary}

\begin{proof}
    For a tensor $T=[T_1,\dots,T_p]\in\bC^{n\times n \times p}$ with $T_1$ invertible $\Rank(T)\geq \Rank(T_1)=n$. Also, $\Rank(T)\leq n$ iff there exists matrices $A\in\bC^{n\times n},B\in\bC^{n\times n}$ and diagonal matrices $D_1,\dots, D_p \in \bC^{n\times n}$ such that $T_i=AD_iB$ for all $i\in[p]$. Since $T_1=AD_1B$ is invertible, so are $A,D_1,B$. The above criterion is then equivalent to the   existence of $A\in GL_n(\bC)$ such that $T_iT_1^{-1}=AD_iD_1^{-1}A^{-1}$ for all $i\in[p]$. 

    For the second part of the corollary, note that~(\ref{eq:simdiagcor}) implies that $T_k=UD_kU^{-1}T_1=UD_kV^T$ for $k=2,\ldots,p$. Moreover, the equality $T_k=UD_kV^T$ is also valid for $k=1$ if we define $D_1=I_n$. The conclusion therefore follows from Proposition~\ref{th:slices}.
\end{proof}

\begin{lemma}\label{lem: interpolation}
    Let $T\in\bC^{m\times n \times p}$ such that $\Brank(T)=r$. Then $\Rank(T)\leq (\edeg(T)+1)r$, where $\edeg(T)$ is the error-degree of $T$.
\end{lemma}

\begin{proof}
    Let $d:=\edeg(T)$. From the definition of \textit{error-degree} (\cref{def:errordegree}), there exists $T(\e)\in\bC[\e]^{m\times n\times p}$ such that $T(\e)= \sum_{i=0}^d T_i\e^d$ for some tensors $T_i\in\bC^{m\times n\times p}$ and $T_0=T$. Choose $\alpha_0,\dots, \alpha_d\in\bC$ such that $\Rank(T(\alpha_i))\leq r$ for all $0\leq i\leq d$. From interpolating $T(\e)$ on $\alpha_0,\dots,\alpha_d$, we get that $T=T_0=\sum_{i=0}^d a_iT(\alpha_i)$ for some $a_i\in\bC$. Consequently, $\Rank(T)\leq \sum_{i=0}^d \Rank(T(\alpha_i))\leq (d+1)r$, which concludes our proof.
\end{proof}

\begin{remark}
    The above lemma gives us a way to get upper bounds on the rank of tensors from upper bounds on error-degrees. This is the version of \cite[Thm. 46]{dutta2025recent} for tensors.
\end{remark}

Error degree and and order of border degeneration are connected by the following inequality:

\begin{lemma} \label{lem:errorfromdegeneration}
    For any $T\in \bC^{m\times n \times p}$, we have $\edeg(T) \leq  2 \bdorder(T)$.
\end{lemma}

\begin{proof}
   Suppose that the vectors $u_i(\epsilon) \in \bC[\e]^m$, $v_i(\epsilon) \in \bC[\e]^n$ and  $w_i(\epsilon) \in  \bC[\e]^p$ are such that  
   \begin{linenomath}
       $$\sum_{i=1}^r u_i(\e)\tp v_i(\e)  \tp w_i(\e) = \e^k T + \e^{k+1}Q$$ 
   \end{linenomath}
   for some $Q\in \bC[\e]^{m\times n \times p}$, 
   %and $u(\e),v(\e),w(\e)\in \bC[\e]$, 
   where $r=\Brank(T)$ and $k=\bdorder(T)$. Then  there exists $Q'\in \bC[\e]^{m\times n \times p}$ and vectors  $u'_i(\epsilon) \in \bC[\e]^m$, $v'_i(\epsilon) \in \bC[\e]^n$,  $w'_i(\epsilon) \in  \bC[\e]^p$ 
   with entries of degree of at most $k$
   such that $\sum_{i=1}^r u'_i(\e)\tp v'_i(\e) \tp w'_i(\e) = \e^k T + \e^{k+1}Q'$. Then, $deg(\e^k T+\e^{k+1}Q')\leq 3k$. Let $T(\e):= T + \e Q'$. Then, $deg(T(\e))\leq 2k$ and $\Rank T(\e) \leq r$ for all $\e \neq 0$.
\end{proof}

\begin{remark} \label{rem:2bdorder}
The upper bound $2\;\bdorder(T)$ in the above lemma is actually  an upper  bound on the degree of error as defined in \cite{dutta2025recent} too, since $\Rank T(\e) \leq r$ for all $\e \neq 0$.
\end{remark}
With the above remark we can obtain bounds on the degree of error as defined in \cite{dutta2025recent} from bounds on $\bdorder(T)$. 
Definition~\ref{def:errordegree} captures most of the spirit of the definition from~\cite{dutta2025recent} but has the following advantage: it sometimes allows us to give better bounds, with a simpler proof, than if we go through the route of Lemma~\ref{lem:errorfromdegeneration} and Remark~\ref{rem:2bdorder}. See Theorem~\ref{thm: not_1-regular}, where we give a smaller bound on the error of degree 
(in the sense of Definition~\ref{def:errordegree}) than on the order of border degeneration, with a simpler proof. We leave it as an open problem whether these two notions of error degree are really distinct (the worse bound obtained with the definition from \cite{dutta2025recent} might be only an artifact of our proof method).

Observe that \cref{lem: interpolation} and \cref{lem:errorfromdegeneration} give us $\Rank(T)\leq (2\bdorder(T)+1)\Brank(T)$. This is exactly the upper bound we get from \cite[Prop 15.26]{BCS97}.

\begin{lemma} \label{lem:mult}
Let $A \in GL_m(\bC),B\in GL_n(\bC)$ and let $T \in \bC^{m\times n \times p}$ be a tensor with slices $T_1,\ldots,T_p$.
If $T' := (A \otimes B^t \otimes I_p) \cdot T$, then $\edeg(T) = \edeg(T')$ and $\bdorder(T) = \bdorder(T')$.
\end{lemma}

\begin{proof}
 By \cref{th:slices}, $\Rank(T') \leq \Rank(T)$.
The slices of $T$ can be obtained from those of $T'$ by multiplication by $A^{-1}$ on the left and $B^{-1}$ on the right, 
hence $\Rank(T) \leq \Rank(T')$ as well.

For border rank, suppose $T = \lim_{\e \rightarrow 0} T(\e)$ where $\Rank T(\e) \leq r$
for every small enough $\epsilon$. Then $T'=\lim_{\e \rightarrow 0} T'(\e)$, 
where the slices of $T'(\e)$ are obtained from those of $T(\epsilon)$ by multiplication by $A$ on the left and $B$ on the right. Since $\Rank(T(\e)) = \Rank(T'(\e))$, we have $\Brank(T') \leq r$. A similar argument shows that $\Brank(T) \leq \Brank(T')$.

Continuing with the definitions of $T(\e)$ and $T'(\e)$ as above, Let $T(\e):=\frac{1}{\e^q}\sum_{i=1}^r u_i(\e)\tp v_i(\e)\tp w_i(\e)$ where $q=\bdorder(T)$ and $u_i\in\bC[\e]^m,v_i\in\bC[\e]^n,w_i\in\bC[\e]^p$. 
Then, $T'(\e)=\frac{1}{\e^q}\sum_{i=1}^r (Au_i(\e))\tp (B^t v_i(\e))\tp w_i(\e)$ and therefore $\bdorder(T')\leq \bdorder(T)$. 
Similarly, $T(\e)=(A^{-1}\tp (B^t)^{-1}\tp I_p)\circ T'(\e)$ and hence $\bdorder(T)\leq \bdorder (T')$.

Now, suppose $T(\e)=[T_1(\e),\dots,T_p(\e)]$. Then, $T'(\e)=[T'_1(\e),\dots, T'_p(\e)]$ where $T'_i(\e)=AT_i(\e)B$ for $i\in[p]$. 
Therefore, 
$\edeg(T') \leq \edeg(T)$.
The converse follows again from the invertibility of $A$ and $B$.
\end{proof}

\begin{remark}
% \label{rem: border_decomp}
    From \cref{lem:mult}, it suffices to find the error-degree and degeneration order for the tensor $T':=[I,T_2T_1^{-1},\dots, T_pT_1^{-1}]$ because $T=(I\tp T_1^t\tp I)\circ T'$ and hence $T$ and $T'$ have the same error-degree and degeneration order. 
    So, without loss of generality we assume $T_1=I$. 
\end{remark}

The following lemma will be needed in Section~\ref{sec:1reg}. 
It deals with tensors that have slices with a block-diagonal structure. 
\begin{lemma}\label{lem: block_diag}
    Let $U,V,W$ be vector spaces over $\bC$ and $T \in U\tp V\tp W$ be a tensor such that $T=S+Q$ where $S\in U_1\tp V_1\tp W$ and $Q\in U_1^{\perp}\tp V_1^{\perp}\tp W$. 
    Suppose $\Brank(T)=\Brank(S)+\Brank(Q)$. 
    Then, $\edeg(T)\leq max\left( \edeg(S),\edeg(Q)\right)$ and $\bdorder(T)\leq max\left( \bdorder(S),\bdorder(Q)\right)$.
\end{lemma}

\begin{proof}
    Let $e_S=\edeg(S),e_Q=\edeg(Q),q_S=\bdorder(S),q_Q=\bdorder(Q)$. 
    Now, let 
    \begin{linenomath}
        $$S(\e)=\frac{1}{\e^{q_S}}\sum_{i=1}^{r_1} u_i(\e)\tp v_i(\e)\tp w_i(\e) \text{ and }  Q(\e)=\frac{1}{\e^{q_Q}}\sum_{i=1}^{r_2} u_i'(\e)\tp v_i'(\e)\tp w_i'(\e)$$
    \end{linenomath} 
    be such that $\lim_{\e\rightarrow 0}S(\e)=S$ and $\lim_{\e\rightarrow 0}Q(\e)=Q$ where $r_1=\bdorder(S),r_2=\bdorder(Q)$ and the entries of $u_i(\e),v_i(\e),w_i(\e),u'_i(\e),v'_i(\e),w'_i(\e)$ are in $\bC[\e]$. 
    Note that $S(\e)+Q(\e)$ is a border decomposition for $T$. Suppose $q_S \geq q_Q$ without loss of generality. We then get  
    \begin{linenomath}
        $$S(\e)+Q(\e)=\frac{1}{q_S}\left( \sum_{i=1}^{r_1} u_i(\e)\tp v_i(\e)\tp w_i(\e) + \e^{q_S-q_Q}\sum_{i=1}^{r_2} u'_i(\e)\tp v'_i(\e)\tp w'_i(\e)\right)$$    
    \end{linenomath}
    Therefore, $\bdorder(T)\leq max\left( \bdorder(S),\bdorder(Q)\right)$.

    For the bound on error-degree, suppose $S(\e),Q(\e)$ are such that $S(0)=S,Q(0)=Q$ and the entries of $S(\e)$ and $Q(\e)$ has $\e$-degree at most $e_S$ and $e_Q$ respectively. Then, again $S(\e)+Q(\e)$ is a border decomposition for $T$ and hence $\edeg(T)\leq max\left( \edeg(S),\edeg(Q)\right)$.
\end{proof}

For the next lemma, recall that $\nu_\e$ denotes the valuation defined in~\cref{sec:valuation}. 
\begin{lemma}\label{lem: border_bound}
    Let $T\in \bC^{m\times n \times p}, T\neq 0$ be a tensor with $\Brank(T)=r$. 
    Assume that $lim_{\e\rightarrow 0}T(\e)=T$ where
    $T(\e):=\sum_{i=1}^r u_i(\e)\tp v_i(\e) \tp w_i(\e)$ for some $u_i(\e)\in\bC(\e)^m,v_i(\e)\in\bC(\e)^n, w_i(\e)\in\bC(\e)^p$.
    %for all $i\in [r]$. 
    Then, $\bdorder(T)\leq -min_{i,j}(\nu_\e(u_i(\e)_j))-min_{i,j}(\nu_\e(v_i(\e)_j)) - min_{i,j}(\nu_\e(w_i(\e)_j))$.
\end{lemma}

\begin{proof}
    For every $i\in [r],j\in [m]$ we do the following. Let $u_i(\e)_j:= f(\e)/g(\e)$ for some polynomials $f(\e),g(\e)\in\bC[\e]$ with $f(\e),g(\e)$ coprime and without loss of generality let the coefficient of the term of $g(\e)$ with the least degree and having a non-zero coefficient be equal to $1$. Define $u'_i(\e)_j:=f(\e)/\e^{\nu_\e(g(\e))}$ and observe that $\nu_\e(g(\e))\geq -\nu_\e(u_i(\e)_j)$. Similarly define $v'_i(\e)_j$ and $w'_i(\e)_j$ for all $i,j$ ($j\in [n]$ and $j\in [p]$ respectively). Now define $u'_i(\e)\in\bC(\e)^m,v'_i(\e)\in\bC(\e)^n, w'_i(\e)\in\bC(\e)^p$ as $u'_i(\e):=(u'_i(\e)_j)_{j\in[m]}$, $v'_i(\e):=(v'_i(\e)_j)_{j\in[n]}$ and $w'_i(\e):=(w'_i(\e)_j)_{j\in[p]}$. And now, let $T'(\e):=\sum_{i=1}^r u'_i(\e)\tp v'_i(\e) \tp w'_i(\e)$. Then, $lim_{\e\rightarrow 0}T'(\e)=lim_{\e\rightarrow 0}T(\e)=T$ which gives us $T'(\e)=T+\e Q(\e)$ for some tensor $Q(\e)\in \bC[\e]^{m\times n \times p}$. Let $t_u:=-min_{i,j}(\nu_\e(u_i(\e)_j)), t_v:=-min_{i,j}(\nu_\e(v_i(\e)_j)), t_w:=-min_{i,j}(\nu_\e(w_i(\e)_j))$ and $t=t_u + t_v +t_w$. Note that $\e^t T'(\e)=\sum_{i=1}^r \e^{t_u}u'_i(\e)\tp \e^{t_v}v'_i(\e) \tp \e^{t_w}w'_i(\e)= \e^t T + \e^{t+1}Q(\e)$ with $\e^{t_u}u'_i(\e)\in\bC[\e]^m, \e^{t_v}v'_i(\e)\in\bC[\e]^n$ and $\e^{t_w}w'_i(\e)\in\bC[\e]^p$. % If $t\leq0$, 
    If $t <0$ then $lim_{\e\rightarrow 0} T'(\e)=0$ because the terms of $T'(\e)$ would be polynomials in $\e$ with strictly positive valuation (a contradiction with the hypothesis $T \neq 0$). So, $t$ has to be non-negative. Therefore, 
    \begin{equation*}
        \bdorder(T)\leq t = -min_{i,j}(\nu_\e(u_i(\e)_j))-min_{i,j}(\nu_\e(v_i(\e)_j)) - min_{i,j}(\nu_\e(w_i(\e)_j)). \qedhere
    \end{equation*}
\end{proof}

%================================================================================
\subsection{Dimension reduction for undercomplete tensors}
%================================================================================

In this subsection we establish some technical results on the reduction of the dimensions of the ambient tensor space whenever our tensor is (border) undercomplete.
In the broader context of arithmetic circuits, this is related to the notion of ``essential variables'' of a polynomial~\cite{Carlini06,Kayal11}.
In the specific context of set-multilinear arithmetic circuits (i.e., tensors), see Lemma 5.1 in the arxiv version of~\cite{BSV21} (where dimension reduction is called ``width reduction''). The results in this subsection are based on standard techniques but we could not find exactly the same statements in the literature, so we provide complete 
proofs.

We first introduce the notion of rank factorization of a matrix. 

\begin{definition}
    Given an $m\times n$ matrix $A$ of rank at most $r$, $A=BC$ is a $r$-\textit{factorization} where $B$ is of size $m\times r$ and $C$ is of size $r\times n$.
\end{definition}

Every matrix $A$ has a $\Rank(A)$-\textit{factorization} (which is simply called a rank factorization of $A$) because if $A=\sum_{i=1}^r u_i\tp v_i$, then $A=UV^t$ where $U:=\{u_1,\dots, u_r\},V:=\{v_1,\dots,v_r\}$. From a similar line of reasoning, we can also say that every matrix $A$ has a $r$-\textit{factorization} for $r\geq \Rank(A)$. 

\begin{remark}
    Let $A$ be a matrix of rank $r$. If $A=BC$ is a $r$-\textit{factorization} of $A$, then $B$ and $C$ have to be matrices of full rank because if not, $A$ has a decomposition into a sum of $k$ rank one matrices for some $k<r$.
\end{remark}

\begin{lemma}\label{lem: special_rk_factor}
    Let $A \in M_{m \times n}(\bC)$. For all $r$ such that  
    $\Rank(A)\leq r$, $A$ has a $r$-factorization $A=A_1A_2$ where $A_1$ has full rank.
\end{lemma}
\begin{proof}
    Let $A=BC$ be some $r$-factorization of $A$. Suppose the $m\times r$ matrix $B$ has rank $s<r$ (If s=r, we have nothing to prove). We can get a $s$-\textit{factorization} of $B$, $B=DE$ for some matrices $D\in\bC^{m\times s}$ and $E\in\bC^{s\times r}$. Then there is a full rank matrix $D'\in\bC^{m\times r}$ such that $D'$ restricted to the first $s$ columns is the matrix $D$. We can construct $D'$ from $D$ by extending the set of column vectors of $D$ to a basis of some $r$-dimensional space $V\subseteq\bC^m$. Now, note that the linear map defined by the matrix $E$, say $L_E:\bC^r\rightarrow\bC^s$, has a lift to a linear map $L_{E'}:\bC^r\rightarrow\bC^r$ such that $L_{E'}=\iota \circ L_E$ where $\iota$ is the standard inclusion of $\bC^s\subseteq\bC^r$ which takes a vector $v=(v_1,v_2,\dots,v_s)$ to $\iota(v)=(v_1,v_2,\dots,v_s, 0, \dots, 0)\in\bC^r$. We get the matrix $E'$ corresponding to $L_{E'}$ w.r.t the standard basis by adding $(r-s)$ rows of $0$s to $E$. Then, we have a $r$-\textit{factorization}, $B=D'E'$, of $B$ which gives a $r$-\textit{factorization}, $A=D'(E'C)$, of $A$ with $D'$, a full rank matrix (we take $A_1:=D$ and $A_2:=E'C$).
\end{proof}

\begin{lemma}\label{lem: dim_shift}
    Let $T= [T_1,T_2,\dots,T_p] \in\bC^{m\times  n\times p}$ be such that $\Brank(T)\leq r\leq min(m,n)$. Then there is a matrix $A\in M_{m\times r}(\bC)$ with full rank and a tensor $S= [S_1,S_2,\dots,S_p] \in\bC^{r\times  n\times p}$ such that $T= (A\tp I\tp I)S$.
\end{lemma}
\begin{proof}
    Since $\Brank(T)\leq r$, there exists $T(\e)$ such that $lim_{\e\rightarrow 0}T(\e)=T$ and $T_i(\e)=U(\e)D_i(\e)V(\e)^t$ for some matrices $U(\e),V(\e),D_i(\e)$ for all $i\in[p]$, where $U(\e),V(\e)$ are rank $r$ matrices and $D_i$ is a size $r$ diagonal matrix for all $i\in [p]$. Let $\widehat{T}(\e):= (T_1(\e),\dots,T_p(\e))$ be the $m \times np$ matrix of slices of $T(\e)$.
    We also define the $n \times mp$ matrix: $\widehat{T^t}(\e):= (T_1(\e))^t,\dots,(T_p(\e))^t)$.

    The rank of $\widehat{T(\e)}$ is at most $r$ because $U(\e)$ has rank at most $r$ and $im(\widehat{T}(\e))\subseteq im(U(\e))$. Therefore, $lim_{\e\rightarrow 0}\widehat{T}(\e)=\widehat{T}$ also has rank at most $r$.
    %$r<m$ which implies that $\Rank(\widehat{T})\leq r$. 
    Due to this and \cref{lem: special_rk_factor}, there is a rank $r$ matrix $A$ of shape $m\times r$ and matrices $S_i$ such that $AS_i=T_i$ for all $i\in [p]$.
\end{proof}

\begin{corollary}\label{cor: dim_shift}
    Let $T= [T_1,T_2,\dots,T_p] \in\bC^{m\times  n\times p}$ be such that $\Brank(T)\leq r\leq min(m,n)$. Then each $T_i$ can be written as $T_i=AZ_iB$ for matrices $A\in M_{m\times r}(\bC)$ and $B\in M_{r\times n}(\bC)$, $rk(A)=rk(B)=r$.
\end{corollary}

\begin{proof}
    Using \cref{lem: dim_shift}, we get a matrix $A$ and a tensor $S$ satisfying the conditions of the lemma. Note that $\Brank(S)\leq r$ because $\Brank(T)\leq r$. We can then use the same lemma with respect to the second tensor coordinate instead on the tensor $S\in \bC^{r\times n \times p}$ to get a matrix $B'\in M_{n\times r}(\bC)$ with full rank and a tensor $Z= [Z_1,Z_2,\dots,Z_p] \in\bC^{r\times  r\times p}$ such that $S= (I\tp B'\tp I)Z$. Setting $B:=(B')^t$, we are done.
\end{proof}

\begin{proposition}\label{thm: undercomplete}
    Let $T\in [T_1,T_2,\dots,T_p] \in\bC^{m\times n\times p}$. Then $\Brank(T)\leq r\leq min(m,n)$ if and only if there exists rank $r$ matrices, $A\in M_{m\times r}(\bC),B\in M_{r\times n}(\bC)$ such that $T_i=AZ_iB$ for some $Z_i\in M_r(\bC)$ for all $i\in[p]$ and the tensor $Z:=[Z_1, Z_2, \dotsm Z_p]$ has $\Brank(Z)\leq r$
\end{proposition}

\begin{proof}
    If $\Brank(T) \leq r$, \cref{lem: dim_shift} shows that there are matrices $A,B$ and $Z_i$ such that $T_i = A Z_i B$ for all $i \in [p]$. 
    Since $A,B^t$ have rank $r$, there exists matrices $A^+, (B^t)^+$ such that $A^+A = I_r$ and $(B^t)^+B^t=I_r$. 
    We will take  $A^+, (B^t)^+$ equal to the Moore-Penrose inverses
    of $A$ and $B^t$. This is a convenient choice because the Moore-Penrose inverse of a matrix $M$ satisfies the property $(M^t)^+=(M^+)^t$. In particular, $((B^t)^+)^t=B^+$ is a right inverse of $B$.
    
    Given a tensor $T(\e)$ converging to $T$, we can construct a tensor $Z(\e):=(A^+\tp (B^t)^+\tp I)(T(\e))$. Then $lim_{\e\rightarrow 0} Z(\e)=Z$ since $Z=(A^+\tp (B^t)^+\tp I)(T)$. Here, we used the fact that $((B^t)^+)^t$ is a right inverse of $B$. We conclude that $\Brank(Z)\leq \Brank(T)\leq r$, which proves one direction of the theorem.

    For the other direction, suppose $T_i=AZ_iB$ for some $A\in M_{m\times r}(\bC),B\in M_{r\times n}(\bC)$ and $Z_i\in M_r(\bC)$ for all $i\in[p]$ and the tensor $Z:=[Z_1, Z_2, \dots, Z_p]$ has $\Brank(Z)\leq r$. Then $T= (A\tp B^t\tp I)(Z)$. So again using the same argument as before, suppose $Z(\e)$ is a rank $r$ tensor over $C(\e)$ that converges to $Z$. Then, $T(\e):=(A\tp B^t\tp I)(Z(\e))$ converges to $T$ as $\e\rightarrow 0$. Furthermore, since $rk(T(\e))\leq rk(Z(\e))\leq r$, we have $\Brank(T)\leq \Brank(Z)\leq r $.
\end{proof}

\begin{corollary}\label{cor: undercomplete decompositions}
    Let $T\in [T_1,T_2,\dots,T_p] \in\bC^{m\times n\times p}$. Then $\underline{rk}(T)= r\leq min(m,n)$ if and only if there exists rank $r$ matrices, $A\in M_{m\times r}(\bC),B\in M_{r\times n}(\bC)$ such that $T_i=AZ_iB$ for some $Z_i\in M_r(\bC)$ for all $i\in[p]$ and the tensor $Z:=[Z_1, Z_2, \dotsm Z_p]$ has $\Brank(Z)= r$. Furthermore:
    \begin{itemize}
    \item[(i)] $Z(\e)$ is a tensor approximating $Z$ and $\Rank(Z(\e))=r$ for $\e$ small enough if and only if  $T(\e):= (A\tp B^t\tp I)(Z(\e))$ is a tensor approximating $T$ with $rank(T(\e))=r$ for small enough $\e$.
    \item[(ii)] $\edeg(T) = \edeg(Z)$ and $\bdorder(T) = \bdorder(Z)$.
    \end{itemize}
\end{corollary}

\begin{proof}
    To first prove the $(\then)$ direction, suppose $\Brank (T)=r$. Then from \cref{thm: undercomplete}, we get a tensor $Z=[Z_1, \dots , Z_p]$ of shape $r\times r \times p$ such that $T_i=AZ_iB$ for all $i\in[p]$ where $A$ and $B$ are full rank matrices of size $m\times r$ and $r\times n$ respectively. Moreover, $\Brank (Z)\leq r$.
    If $\Brank (Z)\leq r-1$, then there exists some tensor $Z(\e)$ converging to $Z$ with $Z(\e)$ has rank $r-1$. This gives us a tensor $T(\e):= (A\tp B^t\tp I)(Z(\e))$ that converges to $T$ which has rank $r-1$. This is a contradiction since $\Brank (T)=r$. Therefore, $\Brank (Z)=r$ too.

    To prove the % $(\Leftarrow)$ 
    converse direction, given $T$ and $Z$ satisfying the conditions of the theorem, observe that $T= (A\tp B^t\tp I)(Z)$ implies $Z = (A^+\tp (B^t)^+\tp I)(T)$ where $A^+$ and $(B^t)^+$ are the  Moore-Penrose inverses of $A$ and $B^t$ respectively. Then, from the same argument as in the $(\then)$ direction, it follows that $\Brank(T)=r$ since $\Brank(Z)=r$.

    Now for the second part of the corollary, for any $T(\e)$ and $Z(\e)$ approximating $T$ and $Z$ such that $T(\e)= (A\tp B^t\tp I)(Z(\e))$, $rank(T(\e))=rank(Z(\e))$ for all $\e$ small enough because $A$ and $B^t$ are full rank matrices (of rank $r$). Therefore, item (i) follows.

    Let $k=\edeg(Z)$ and let $Z(\epsilon)$ be a corresponding degree $k$ approximating tensor. Then from~(i) we have that 
    $T(\e)=(A\tp B^t\tp I)(Z(\e))$ is an approximating tensor for $T$
    of degree at most $k$, with $rank(T(\e))=r$ for small enough $\e$.
    This shows that $\edeg(T) \leq \edeg(Z)$, and likewise we have 
    $\edeg(Z) \leq \edeg(T)$ since $Z = (A^+\tp (B^t)^+\tp I)(T)$.
    A similar argument shows that $T$ and $Z$ have same order of degeneration.
\end{proof}

%================================================================================
\section{Approximately Simultaneously Diagonalizable (ASD) matrices}
\label{sec:asd}
%================================================================================

\begin{definition}
    $A_1,A_2,\dots, A_m\in M_n(\bC)$ are said to be \textit{approximately simultaneously diagonalizable (ASD)} if for every $\e>0$, there exist simultaneously diagonalizable matrices $A_i(\e), i\in[m]$ such that $\norm{A_i(\e)-A_i}<\e$ for each $i\in [n]$ where we take the norm to be the Frobenius norm.\footnote{We could have used other matrix norms, but we chose the Frobenius norm for simplicity (the ASD property is in fact independent of the choice of the norm).}
\end{definition}

\begin{proposition}\label{thm: brank_n_char}
    For $T = [T_1,\dots,T_p] \in\bC^{n\times n \times p}$ with $T_1$ invertible, we have $\Brank(T)= n$ iff the matrices $T_2T_1^{-1},\dots, T_pT_1^{-1}$ are approximately simultaneously diagonalizable.
\end{proposition}

\begin{proof}
    Let $T\in\bC^{n\times n \times p}$ be a tensor with $T_1$ invertible. 
    Notice that for any  $T(\e)$ such that $lim_{\e\rightarrow 0}T(\e)=T$, $T_1(\e)$ is invertible and hence $\Rank(T(\e))\geq n$ for all but finitely many $\e$. Therefore, $\Brank(T)\geq n$. Note that for our tensors $T$, $\Brank(T)\leq n$ iff $\Brank(T)=n$ because $T_1$ is invertible.
    
    Now, $\Brank(T)= n$ implies that there exists $T(\e)\in\bC[\e]^{m\times n \times p}$ with $lim_{\e\rightarrow 0}T(\e)=T$ and $\Rank(T(\e))=n$ for all but finitely many $\e$. But, from \cref{cor: rank_n_char}, this implies that the matrices $T_i(\e)T_1(\e)^{-1},\; i\in[p]$, are simultaneously diagonalizable. Also $lim_{\e\rightarrow 0}T_i(\e)T_1(\e)^{-1}=T_iT_1^{-1}$. Hence, the matrices $T_iT_1^{-1},\; i\in[p]$, are ASD.

    For the other direction, suppose $T_2T_1^{-1},\dots, T_pT_1^{-1}$ are ASD. 
    Let $A_i:=T_iT_1^{-1}$. 
    By the equivalence between topological degeneration and algebraic degeneration (\cite[Sections 20.6 and 20.7]{BCS97}), there exists an $\e$-family of matrices, $A_i(\e)\in\bC[\e]^{m\times n}$, such that $A_i(\e),\; i\in[p]$ are simultaneously diagonalizable for all but finitely many $\e$. 
    From $\cref{cor: rank_n_char}$, this implies that $\Rank(T(\e))=n$ where $T(\e):=[A_1(\e)T_1,A_2(\e)T_1,\dots, A_p(\e)T_1]$. 
    Furthermore, $lim_{\e\rightarrow 0}T(\e)=T$. 
    Therefore, $\Brank(T)\leq n$ which implies $\Brank(T)=n$.
\end{proof}

\noindent We next state the Motzkin-Taussky theorem, from \cite[Theorem 5]{MT1955}.

\begin{theorem}[Motzkin-Taussky]\label{Thm: Motzkin-Taussky}
    Every pair of complex commuting matrices has the ASD property.
\end{theorem}

%===============================================================
\subsection{The \texorpdfstring{$E_{n,1}$}{(n,1)} perturbation}\label{subsection: special perturbation}
%===============================================================

In this subsection, we show that we can get simultaneously diagonalizable perturbations which are linear in $\e$, at the cost of having simultaneous diagonalizations which are not polynomial in $\e$. 
This will help us in later sections to get better bounds on the error degree of tensors, when compared to the degeneration order.

We begin with a simple lemma.

\begin{lemma}\label{lem:matrix_words}
    Let $M,N$ be two matrices with $NM^kN=0$ for $0 \leq k \leq n-2$ (where $M^0 = I$). 
    Then, $\displaystyle (M+N)^i=M^i +\sum_{0 \leq k \leq i-1}M^{k}NM^{i-1-k}$ for all $i\in[n]$. The expression is hence linear in $N$.
\end{lemma}
\begin{proof}
    The expansion of $(M+N)^i$ only contains words of the form $M^{i_1}N^{j_1}M^{i_2}N^{j_2}\dots M^{i_s}N^{j_s}$. 
    But since $NM^kN=0$ and $N^2=0$, the degree of $N$ in these monomials is at most $1$. 
\end{proof}
We will apply this lemma with $N=E_{n,1}$, the elementary basis matrix with a 1 in row $n$ and column 1 (as defined at the beginning of~\cref{sec:prelim}).

% We tackled the 1-regular case in the previous section, but here, we give a degree 1 perturbation that is ASD which comes with the cost of having eigenvalues that are not polynomial or even rational in $\e$ (and as a consequence, the corresponding tensor decomposition might not even have rational terms which we'll see in the next section).

\begin{proposition}\label{prop:1-reg_pert_linear}
    Let $A_1,A_2,\dots, A_p$ be matrices in $M_n(\bC)$ such that $A_i=p_i(A_1)$ for polynomials $p_i(x)\in \bC[x]$ for $2\leq i \leq p$ with $deg(p_i)\leq n-1$. 
    Then, we have simultaneously diagonalizable $\e$-perturbations of $A_1,A_2,\dots, A_p$, namely $A_1(\e), A_2(\e),\dots, A_p(\e)$ with entries being linear in $\e$.
\end{proposition}

\begin{proof}
    We first give a proof for the case where $A_1$ is a single Jordan block.
    
    \textbf{Special Case:} $A_1$ is a single Jordan block of eigenvalue $\alpha\in\bC$.
    
    In this case $A_1=W+\alpha I$. 
    W.l.o.g., we can assume $\alpha=0$ (i.e, it is enough to find a perturbation of $W$ and $A_2$ that simultaneously diagonalizes them since $A_i=p_i(A_1)=p_i(W+\alpha I)=p'(W)$ for a different polynomial $p'$ of degree at most $n-1$).
    
    We then consider the following perturbation: $W(\e) := W + \e \cdot E_{n,1}$. 
    Note that the minimal polynomial of $W(\e)$ is $f_{W(\e)}(x)= x^n+ (-1)^{n-1}\e$. Therefore, $W(\e)$ has distinct roots, which are $\lb_k(\e):= \e^{1/n}e^{
    (2k+1)\iota\pi/n}= (-\e)^{1/n}\omega_n^k$ for $0\leq k \leq n-1$ where $\omega_n=e^{2ki\pi/n}$. Therefore, $W(\e)$ is diagonalizable. Let $A_i(\e):=p_i(W(\e))$ for $i\in[p]$. We are now left to prove that $A_i(\e)$ is linear in $\e$ for all $i\in[p]$. Since $A_i(\e)=p(W+\e N)$ where $N:=E_{n,1}$, and since $N^2=0$ and $NW^kN=0$ for all $k\leq n-2$, from  \cref{lem:matrix_words}, $\deg_\e(A_i(\e)) \leq 1$.

    \textbf{General Case:} $A_1$ is a Jordan matrix with multiple Jordan blocks.

    In this case, since $A_i=p_i(A_1)$ for all $i\in [p]$, the matrices $A_1,\dots A_p$ have a matching block diagonal structure. Suppose $A_1= diag(J_1,J_2,\dots J_s)$ for some Jordan blocks $J_1,\dots,J_s$. Then, $A_i=p_i(A_1)= diag(p_i(J_1),\dots, p_i(J_s))$. From case 1, we know that there exists linear $\e$ perturbations of $J_k$ and $p_i(J_k)$, namely $J_k(\e)$ and $p_i(J_k(\e))$, for all $i$ and any fixed $k\in[s]$ that are simultaneously diagonalizable. Define $A_1(\e):= diag(J_1(\e),\dots, J_s(\e))$ and $A_i(\e):=p_i(A_1(\e))=diag(p_i(J_1(\e)),\dots, p_i(J_s(\e)))$. Then, $A_1(\e),A_2(\e),\dots, A_p(\e)$ are simultaneously diagonalizable perturbations linear in $\e$.
\end{proof}

\section{Degree of Error and Order of Degeneration}\label{sec:border}
%================================================================================

In this section we prove our main border rank and debordering results.
We begin each subsection by first establishing the \say{almost square} tensor case, i.e. where our tensor is 3-generic in $\bC^{n \times n \times p}$ or in $\bC^{n \times n \times 3}$. 
Once we have proven our technical results (\cref{thm:border1reg,thm: not_1-regular}), we show how our main theorems (\cref{thm: general m not n,thm: linear part m not n}) are derived from them.
In \cref{sec:over} we extend our main results to the overcomplete case.

\subsection{The 1-regular case} \label{sec:1reg}

We begin this subsection assuming that we are given a $(n,3)$-generic tensor $T \in \bC^{n \times n \times p}$ where $T_2 T_1^{-1}$ is 1-regular, which is where our technical contribution comes in quite explicitly.
We then deduce \cref{thm: linear part m not n} from our main technical results.

We begin with a special case.

\begin{proposition}\label{prop: linear_part}
    Let $T\in \bC^{n\times n \times p}$ be a tensor such that $T_1$ is invertible. 
    Suppose that $T_iT_1^{-1}$ is a polynomial in $T_2T_1^{-1}$ for all $i\in [p]$.
    %and $T_iT_1^{-1}$, $T_jT_1^{-1}$ commute for all $i,j\in[p]$. 
    Then, $\Brank(T)=n$, $\edeg(T) \leq 1$ and $\bdorder(T)\leq n-1$. 
    Furthermore, $\Rank(T)\leq 2n$.
\end{proposition}

\begin{proof}
    Let $A_i=T_{i}T_1^{-1}$ for all $i\in [p]$.
    By \cref{lem:mult} we can assume that $T_1=I$ without loss of generality and that our tensor is given by $T := [I, A_2, \dots, A_p]$ with $A_2$ in Jordan normal form.
    
    From our hypothesis, let $p_i(x)\in\bC[x]$ be such that $A_i=p_i(A_2)$. 
    By Cayley-Hamilton, we can take $p_i$ with $deg(p_i)\leq n-1$ for all $i\in [p]$.
    Then, \cref{prop:1-reg_pert_linear} applies and we get linear perturbations $A_i(\e) := A_i + \e \cdot B_i=p_i(A_2(\e))$, with $B_i \in \bC^{n \times n}$, that are simultaneously diagonalizable. 
    Let $T(\e):=[I,A_2(\e),A_3(\e),\dots, A_{p-1}(\e)]$. 
    % The matrices $T_2(\e)T_1^{-1},\dots, T_p(\e)T_1^{-1}$ are simultaneously diagonalizable and consequently $T_2T_1^{-1},\dots, T_pT_1^{-1}$ are ASD. 
    Therefore, $\Brank(T)=n$ from \cref{thm: brank_n_char}. 
    Furthermore, $T(\e)$ has entries linear in $\e$, so $\edeg(T) \leq 1$. 
    From \cref{lem: interpolation}, we get $\Rank(T)\leq 2n$. 
    % We will now prove the upper bound on $\bdorder(T)$.

    \paragraph{Bounding the degeneration order:}
    \makeatletter \global\@newlistfalse\makeatother
    % We will prove the upper bound by finding a "good" border decomposition and then use \cref{lem: border_bound}.  
    Since $A_i = p_i(A_2)$, we know that $A_i$ has the same Jordan block structure as $A_2$. 
    Thus, by \cref{lem: block_diag} it is enough to assume that $A_2$ is a simple Jordan block (i.e. $A_2 = \lambda I +  W_{n}$ for some $\lambda \in \bC$), since by the above paragraph the restriction of $T$ to each single simple Jordan block of $A$ with dimension $d_i$ will have border rank $d_i$. 
    
    % For the general case, if $A_2$ is any matrix, there exists $C\in M_n(\bC)$ such that $C^{-1}A_2C$ is in Jordan normal form. From \cref{lem:mult}, the tensors $T$ and $T':= (C^{-1}\tp C^t\tp I_p)\circ T$ have the same degeneration order. Also, $T'_2(T'_1)^{-1}=C^{-1}A_2C$ is in Jordan canonical form. Furthermore, from \cref{lem: block_diag}, we can assume that $C^{-1}A_2C$ is a Jordan block because $A_i=p_i(A_2)$ and hence $A_i$s for $i\in[p]$ have block diagonal structure matching that of $A_2$. Therefore, from the previous cases, $\bdorder(T)=\bdorder(T')\leq n-1$.

    % Firstly, suppose that $A_2$ is a nilpotent Jordan block. 
    Since we are assuming $A_2 = \lambda I + W$, let $\widehat{A_2}(\e)=A_2 + (-\e)^n E_{n,1}$. 
    Hence, $det(xI-\widehat{A_2}(\e))=(x-\lambda)^n - \e^n$. 
    Therefore, the eigenvalues of $\widehat{A_2}(\e)$ are $\lb_k(\e):= \lambda + \omega_n^k \cdot \e$ for $0\leq k \leq n-1$, where $\omega_n=e^{2i\pi/n}$.
    Note that $\widehat{A_2}(\e)$ is diagonalizable since it has $n$ distinct eigenvalues. 
    Let $\widehat{A_i}(\e):=p_i(\widehat{A_2}(\e))$ for $i\in[p]$.

    Let $\alpha \in\bC[\e]$ be one of the eigenvalues of $\widehat{A_2}(\e)$ and $u_\alpha(\e):=(u_1,\dots,u_n)\in\bC(\e)^{n}$ be an eigenvector of eigenvalue $\alpha$. 
    Then, $(\widehat{A_2}(\e)-\alpha I)u_\alpha(\e)=0$. 
    This gives,

    \begin{equation}
        \begin{bmatrix}
            u_2 - (\alpha-\lb) u_1 \\
            u_3 - (\alpha-\lb) u_2 \\
            \vdots \\
            u_n - (\alpha-\lb) u_{n-1}\\
            (-\e)^n u_1 -(\alpha-\lb) u_n
        \end{bmatrix} = \vec{0}
    \end{equation}
    
    Setting $u_1=1$, we get $u_i= (\alpha - \lb)^{i-1}$ for all $i\in[n]$. 
    Therefore, $u_{\lb_k(\e)}=(1, \omega_n^k \e,\dots, (\omega_n^k \e)^{n-1})$. 
    Let $U(\e)=[u_{\lb_1}(\e),\dots,u_{\lb_n}(\e)]$ be the matrix of eigenvectors of $\widehat{A_2}(\e)$. 
    Then, we have 
    \begin{linenomath}
        $$\widehat{A_i}(\e) = p_i(\widehat{A_2}(\e)) = U \cdot p_i( \diag(\lambda_1(\e), \dots, \lambda_n(\e)) ) \cdot U^{-1}$$
    \end{linenomath}
    which by Corollary~\ref{cor: rank_n_char} yields the border-decomposition
    % \footnote{$T(\e)$ does not have a decomposition into vectors $u_i(\e),v_i(\e),w_i(\e)$ with rational entries in $\e$. Therefore, we cannot get a bound on degeneration order from this. But, we can rectify this through a change of parameters from $\e$ to $\e^{1/n}$.} of $T$. 

    \begin{equation}\label{eq: new_pert}
        T(\e)=\sum_{i=1}^n u_i(\e)\tp v_i(\e)\tp w_i(\e)
    \end{equation}

    \noindent where $[v_1(\e) \cdots v_n(\e)] := U(\e)^{-t}$ and $w_i(\e) := [1, \lambda_i(\e), \dots, p_p(\lambda_i(\e))]^t$.
  % \PK{The above computation exploits a connection between eigenvalues / eigenvectors and tensor decompositions. Is this something that we have explained in the paper already? If so, we should probably point to that here.} 
    We then have the following lower bounds on the $\e$-valuations of $u_i(\e),v_i(\e),w_i(\e)$,
    \begin{equation}\label{eq: valuations}
        min_{i,j}(\nu_\e(u_i(\e)_j))\geq 0,\;\;\; min_{i,j}(\nu_\e(v_i(\e)_j))\geq -(n-1),\;\;\; min_{i,j}(\nu_\e(w_i(\e)_j))\geq 0
    \end{equation}

    \noindent \cref{lem: border_bound} together with \cref{eq: valuations} and \cref{eq: new_pert} gives us $\bdorder(T)\leq n-1$.
\end{proof}
%This was the last paragraph of the above proof
% Now, consider the case when $A_2$ is a Jordan block. 
%     Suppose $A_2$ is not nilpotent. 
%     Say $A_2=J+ bI$ for some $b\in\bC$ and some nilpotent Jordan block $J$. Let $J(\e):=J+\e E_{n,1}$, note that $A_2(\e)=J(\e)+bI$. Then, $J(\e)$ and $A_2(\e)$ have the same eigenvectors. After changing parameters from $\e$ to $\e^{1/n}$ as we did previously, we still get $min_{i,j}(\nu_\e(w_i(\e^n)_j))\geq 0$ even when $A_2$ is not nilpotent. 
%     Therefore, $\bdorder(T)\leq n-1$.

\begin{theorem} \label{thm:border1reg}
    Let $T\in \bC^{n\times n \times p}$ be a tensor such that $T_1$ is invertible and $T_2T_1^{-1}$ is 1-regular. 
    Then, $T_iT_1^{-1}$, $T_jT_1^{-1}$ commute for all $i,j\in[p]$ if and only if $\Brank(T)=n$. 
    Furthermore, when $\Brank(T)=n$, $\edeg(T) \leq 1$ and $\bdorder(T)\leq n-1$.
\end{theorem}

\begin{proof}
    Let $A_i=T_{i+1}T_1^{-1}$ for all $i\in [p-1]$. Since $A_1$ is 1-regular, from \cref{prop:centralizer}, we can write $A_i=p_i(A_1)$ for polynomials $p_i(x)\in\bC[x]$ with $deg(p_i)\leq n-1$ for all $i\in [p]$. Then, by \cref{prop: linear_part}, $\Brank(T)=n$, $\edeg(T) \leq 1$ and $\bdorder(T)\leq n-1$. 
    This proves the direction $(\then)$ and the bounds on the error-degree and degeneration order.

    To prove the direction ($\Leftarrow$), suppose $\Brank(T)=n$. 
    Then, there exists $T(\e)=\sum_{i=1}^n u_i(\e)\tp v_i(\e)\tp w_i(\e)$ such that $lim_{\e\rightarrow 0}T(\e) = T$. This implies that there exist matrices $U(\e),V(\e)$ and $D_i(\e)$ for all $i\in p$ such that $T_i(\e)=U(\e)D_i(\e)V(\e)^t$ for each $i\in [p]$.
    %PK change below:
    %and $U(\e), V(\e)$ are invertible and all 
    Here, the $D_i(\e)$ are diagonal matrices and $T_1(\e)$ is invertible for small enough $\e$ since $T_1$ is invertible.
    Then, since $T_i(\e)T_1(\e)^{-1}=U(\e)D_i D_1^{-1}U(\e)^{-1}$, $T_i(\e)T_1(\e)^{-1}$, $T(\e)_jT_1(\e)^{-1}$ commute for all $i,j\in[p]$. 
    Therefore, their limits also commute and thus we've proved the theorem.
\end{proof}

As a corollary, we obtain our second main theorem on the error degree and order of degeneration in the undercomplete case.

\linearErrorMain*

\begin{proof}
    \cref{thm:border1reg} implies $\Brank(Z)=r$, $\edeg(Z) \leq 1$ and $\bdorder(Z)\leq r-1$. 
    % Now, let $T(\e)= (A \tp B^t \tp I)(Z(\e))$. 
    From \cref{cor: undercomplete decompositions} we have $\Brank(T) = \Brank(Z)$, $\edeg(T) = \edeg(Z)$ and $\bdorder(T) = \bdorder(Z)$. 
    % From \cref{cor: undercomplete decompositions}, $\Rank T(\e) = r$ for all $\e\neq 0$ \RO{we don't guarantee this, right? Maybe we can modify \cref{cor: undercomplete decompositions} to say $\edeg$ and $\bdorder$ are preserved?}\PK{I agree! We already have the preservation of $\edeg$ and $\bdorder$ in Lemma 2.23 in a slightly diffenrent setting (when $A$ and $B$ are invertible), but this needs to be done for Corollary 2.33 as well.} and $lim_{\e\rightarrow 0}T(\e) = T$, thus $\edeg(T) \leq 1$. 
    % Furthermore, $\bdorder(T)\leq \bdorder(Z) \leq r-1$.
\end{proof}

\subsection{The case where \texorpdfstring{$T_iT_1^{-1}$}{TiT1inverse}  is not 1-regular} \label{sec:not1reg}
%================================================================================

%PK A version of \cref{lem:1-reg ext} for the case of matrices $A$ in Jordan form instead of the Weyr form appears as part of the proof of \cite[Theorem~7.6.1, p~342]{OMV11}.
In this section, we drop the the assumption that any $T_2 T_1^{-1}$ is 1-regular, in order to prove \cref{thm: general m not n}. 
The price to be paid for this is that we can only work with tensors of format $m \times n \times 3$, because the proof requires the Motzkin-Taussky theorem.

\begin{proposition}\label{thm: not1reg_main}
    Let $T\in\bC^{n\times n\times 3}$ be a tensor such that $T_1$ is invertible.
    Then, $T_2T_1^{-1}$, $T_3T_1^{-1}$ commute if and only if $\Brank(T)=n$.
\end{proposition}
\begin{proof}

    From \cref{Thm: Motzkin-Taussky} we have that $T_2T_1^{-1},T_3T_1^{-1}$ commute iff they are ASD and from \cref{thm: brank_n_char}, $T_2T_1^{-1},T_3T_1^{-1}$ are ASD iff $\Brank(T)=n$ which completes the proof.
\end{proof}

\begin{lemma}\label{lem:1-reg ext}
    Let $A \in \bC^{n \times n}$. 
    There exists $R \in \bC^{n \times n}$ that is 1-regular and commutes with $A$. Furthermore, if $A$ is in Jordan form or in Weyr form, there exists a choice of $R$ that is upper triangular.
\end{lemma}

\begin{proof} 
    We first show that for every Jordan matrix $A$, we have an upper triangular,  1-regular matrix $R$ that commutes with $A$.
    This appears as part of the proof of \cite[Theorem~7.6.1, p~342]{OMV11}.
    We give a proof below for the sake of completeness, and then show that
    this implies the same result for the Weyr form.
    
    Suppose that $A$ has the form, 

    \begin{equation}
        A:=\begin{bmatrix}
            B_1 & 0 & \dots & 0\\
            0 & B_2 & \dots & 0\\
            \vdots & \dots & \dots & \vdots\\
            0 & 0 & \dots & B_s
        \end{bmatrix}
    \end{equation}
    where $B_j$ is a Jordan block of size $i_j$ with eigenvalue $\lb_j$.

    \noindent Define the matrix $M$,

    \begin{equation}
        M:=\begin{bNiceArray}{c|[end=2]c|[start=2, end=3]c|[start=3,end=3]c|[start=5]c}[margin]
            0_{i_1,i_1} & & & &\\
            \Hline[end=2]
             & I_{i_2} & & & \\
             \Hline[start=2,end=3]
             & & 2I_{i_3} & & \\
             \Hline[start=3,end=3]
             & & & \ddots & \\
             \Hline[start=5]
             & & & & (s-1)I_{i_s}
        \end{bNiceArray}
    \end{equation}

    Take $R:=A+cM$ for some fixed $c\in\bC$. Then, for some choice of $c$, $R$ will be a Jordan matrix with distinct eigenvalues $\lb'_j:=\lb_j + (j-1)c$. 
    Therefore $R$ is 1-regular. Since $M$ commutes with $A$, $R$ commutes with $A$ too. 
    Due to our choice of $M$, $R$ is upper-triangular since $A$ is upper triangular.

    Now let $A'$ be the Weyr matrix associated with $A$. Then $A'$ is upper triangular and $A'=P^{-1}AP$ for some permutation matrix $P$ (from \cref{lem: J_to_W}). Then, $R':=P^{-1}RP$ is 1-regular and commutes with $A'$. Furthermore, $P^{-1}RP=P^{-1}(A+cM)P=A'+cP^{-1}MP$. Since $M$ is diagonal, so is $P^{-1}MP$ and therefore $R'$ is upper triangular.
\end{proof}

\begin{theorem}\label{thm: not_1-regular}
    Let $T=[T_1, T_2, T_3]\in \bC^{n \times n \times 3}$ be a tensor such that $T_1$ is invertible and $T_2T_1^{-1},$ $T_3T_1^{-1}$ commute. 
    Then, $\Brank(T)=n$, $T$ has error-degree at most $(n-1)^3 + (n-1)^2$ and the border degeneration order of $T$ is at most $2(n-1)^3 + 3(n-1)^2 + 3(n-1)$.
\end{theorem}

\begin{proof}
    Note that $\Brank(T)=n$ from \cref{thm: not1reg_main}. %
    Now, to find an explicit perturbation $T(\e)$ of $T$, we will reduce to the case where one of the slices is 1-regular. 
    By applying an invertible change of coordinates and by \cref{thm: weyr_upper_triangular}, without loss of generality we can assume that $A:=T_2T_1^{-1}$ is in its Weyr form and $B:=T_3T_1^{-1}$ is upper triangular. From \cref{lem:1-reg ext}, there exists an upper triangular matrix $R$ that is 1-regular and commutes with $A$.

    From \cref{lem:1-reg ext p2}, for all but finitely many $\e\in\bC$, $R+(1/\e) B$ is 1-regular and thus $B'(\e):=B+\e R$ is 1-regular. Since $A$ commutes with both $B$ and $R$, we have that $A$ commutes with $B+\e R$ which is 1-regular in $\bC(\e)$. Therefore $A=p(B'(\e))=p(B+\e R)$ for some polynomial %PK $p\in \ol{\bC(\e)}[x]$,
    $p\in \bC(\e)[x]$ with $deg(p)\leq n-1$ from \cref{prop:centralizer}.

    \noindent \textbf{Bounding the error-degree:} 
    Summarising what we have so far,
    \begin{align*}
     A=p(B'(\e))=p(B+\e R) &= \sum_{i=0}^{n-1} a_i(\e)(B+\e R)^i
    \end{align*}

    We also know that $B+\e R$ is upper triangular. Then, using $\cref{prop: poly_coeff}$ with $d_A=0$ and $d_B=1$, we get that $a_i(\e)=f_i(\e)/g(\e)$ for some $f_i(\e)\in\bC[\e]$ for all $i\in[n-1]$ and $g(\e)\in\bC[\e]$ with $deg_\e(f_i)\leq (n-1)^2d_B^2 + d_A=(n-1)^2$ and $deg_\e (g(\e))\leq n(n-1)d_B=n(n-1)$.
    
    Define $B(\e):=B+\e R + \e g(\e)D$ and $A(\e):=p(B(\e))$ for some diagonal matrix $D\in M_n(\bC)$ such that $R+ g(\e)D$ has distinct diagonal entries. From \cref{lem: UT_diag}, $B(\e)$ is diagonalizable for all but finitely many $\e$. Therefore, $A(\e)$ and $B(\e)$ are simultaneously diagonalizable and we can write
    \begin{equation} \label{eq:simdiag}
        A(\e) = U(\e) D_2(\e) U(\e)^{-1},\;\; B(\epsilon) = U(\e) D_3(\e) U(\e)^{-1}
    \end{equation}
    where $D_2(\e)$, $D_3(\e)$ are diagonal. 
    We will be computing these matrices, and the matrix of eigenvectors $U(\e)$ later to prove the upper bound on the degeneration order the same way we did in \cref{prop: linear_part}.
    
    We can then define $T(\e):= [T_1, A(\e)T_1, B(\e)T_1]$. Note that for all but finitely many $\e$ we have $T(\e)$ has rank $n$, due to the matrices $I,A(\e),B(\e)$ being simultaneously diagonalizable.\footnote{This is where we use our notion of error degree, see \cref{def:errordegree}.}
    Note that $A(\e)= p(B(\e))$ can be expanded as follows:
    % \begin{align}
    %     B(\e)^i &= \sum_{k=0}^i \binom{i}{k}(B+\e R)^k(\e g(\e)D)^{i-k} \nonumber\\
    %             &= (B+\e R)^i + \e g(\e)D\left( \sum_{k=0}^{i-1}\binom{i}{k}(B+\e R)^k(\e g(\e)D)^{i-k-1}\right) \label{eq: expansion}
    % \end{align}
    \begin{equation}\label{eq: expansion}
        B(\e)^i = \sum_{k=0}^i \binom{i}{k}(B+\e R)^k(\e g(\e)D)^{i-k} 
        = (B+\e R)^i + \e g(\e)D Q_i(\e) 
    \end{equation}

    % Taking $Q_i(\e):= \sum_{k=0}^{i-1}\binom{i}{k}(B+\e R)^k(\e g(\e)D)^{i-k-1}$, we have the following equality,
    Where $Q_i(\e) \in M_n(\bC[\e])$ satisfies $\deg_\e(Q_i(\e))\leq (n-2)*\deg_\e (\e g(\e))\leq (n(n-1)+1)(n-2)$.
    Thus, we have:
    \begin{align}
        p(B(\e)) &= \sum_{i=0}^{n-1} a_i(\e) B(\e)^i = \sum_{i=0}^{n-1} a_i(\e) \left( (B+\e R)^i + \e                                                  g(\e)D*Q_i(\e) \right) \nonumber\\
         &= p(B+\e R) + \sum_{i=0}^{n-1}a_i(\e)g(\e)\left(\e D*Q_i(\e)\right) \nonumber\\
        &= A +\e D \sum_{i=0}^{n-1} a_i(\e)g(\e)Q_i(\e) \label{eq: expansion_2}
    \end{align}
    
    Note that $a_i(\e)g(\e)\in \bC[\e]$ and $Q_i(\e)\in M_n(\bC[\e])$. 
    Therefore $T(\e) \in M_n(\bC[\e])$.
    From the above expressions, we also have $\deg_\e(a_i(\e)g(\e)) = \deg(f_i(\e)) \leq (n-1)^2$.
    Therefore,
    \begin{align}
        deg(p(B(\e))) &\leq 1 + \max_i\left(\deg_\e((a_i(\e)g(\e))\right) + \max_i(\deg_\e(Q_i(\e)))\\
                      &\leq 1 + (n-1)^2 + (n(n-1)+1)(n-2) \nonumber\\
                      &= 1 + (n-1)^2 + \left((n-1)^2 + (n-1) +1\right)\left((n-1) - 1\right) \nonumber\\
                      &= (n-1)^3 + (n-1)^2
    \end{align}

    % To guarantee convergence of $A(\e)$ to $A$, note that we need $k\geq \nu_\e(g(\e)) +1=s+1$ because $A(\e)= A + \e^k(r(\e))Q$ \RO{explain this - this is an explicit calculation} for some rational function $r(\e)$ with the denominator $g(\e)$ and $Q\in\bC^{n\times n}$ \RO{$Q \in \bC[\e]^{n \times n}$}. 
    % We can choose $k=s+1$. 
    % Then the entries of $A(\e)$ will be rational in $\e$ with numerator degree at most \RO{missing a $+ (n-1)^2$ from the coefficients of $a_i(\e)$?}\SG{Yes, I think a "$+ (n-1)^2$" was missing which I've added}$k(n-1) + (n-1)^2= (s+1)(n-1) + (n-1)^2\leq (n+1)(n-1)^2 + n-1$ and denominator $g(\e)$. 
    % Since we also know that the entries of $A(\e)$ are just polynomials in $\e$ \RO{not really, right? Becasue now we have the perturbation by $\e^k D$}, the degree of these polynomial entries in $\e$ will be at most $(n+1)(n-1)^2 + n-1 -deg_\e(g(\e))\leq (n-1)^3 + 2(n-1)^2 + (n-1)$. \SG{I made a change in the final expression that we get from the last inequality (the LHS of the inequality is still the same)}\RO{only can show $\leq n^3$?}. %\PK{At this point, we should note that we have already proved the bound in the theorem's statement on the error degree as defined in~\cite{dutta2025recent}}
    This completes the proof of the upper bound on the degree of error.
    % So far we have not used the fact that $B$ and $R$ are upper triangular, 
    % but these two properties will be used in the remainder of the proof (see also Remark~\ref{rem:weyr}).

   \noindent \textbf{Bounding the border degeneration order:} Since $B(\e):= B+\e R +\e g(\e) D$ is upper triangular, its eigenvalues are polynomials in $\bC[\e]$ with degree at most $k \leq \deg_\e(g(\e)) + 1 \leq n(n-1)+1$. Let the eigenvalues of $B(\e)$ be $\lb_1(\e),\dots \lb_n(\e)$. %PK (which are also the diagonal entries of $B(\e)$ in the same order as listed). 
   These are also the diagonal entries of $B(\e)$ in the same order as listed. From \cref{prop: k-nice_ev}, $B(\e)$ has $n$ distinct eigenvectors and its matrix of eigenvectors $U(\e)$ is $k$-nice, which implies that $U(\e)^{-1}$ is also $k$-\textit{nice}.
   Recall that we defined $T(\e)$ as the tensor $[T_1, A(\e)T_1, B(\e)T_1]$. So far we have simultaneously diagonalised the matrices $T_i(\e)T_1^{-1}= U(\e)D_i(\e)U(\e)^{-1}$ for $i\in [3]$: see~\cref{eq:simdiag} for $i=2,3$, and for $i=1$ we have $D_i(\e)=I_n$. %and $D_i(\e)$ being the matrix with eigenvalues of $T_i(\e)T_1^{-1}$ in its diagonal, giving us 
   Hence $T_i(\e)=U(\e)D_i(\e)U(\e)^{-1}T_1$. 
   Choosing $V(\e):=T_1^tU(\e)^{-t}$ as per 
   Corollary~\ref{cor: rank_n_char}, we obtain $T_i(\e):=U(\e)D_i(\e)V(\e)^t$ which gives us a valid border decomposition. 
   From what we have proved so far, the entries of $U(\e),V(\e)$ are rational functions in $\e$ with valuation at least $(-k(n-1))$.

    % Recall that the diagonal entries of $D_i(\e)$ are exactly the eigenvalues of $T_i(\e)T_1^{-1}$ corresponding to the eigenvectors represented by the columns of $U(\e)$. The matrix $T_i(\e)T_1^{-1}$ is equal to $A(\e)$ when $i=2$ and $B(\e)$ when $i=3$. 
    Also, $A(\e)=p(B(\e))$, hence $D_2(\e)=p(D_3(\e))$. Therefore, the matrix of vectors $W(\e):=[w_1(\e), w_2(\e),\dots , w_n(\e)]$ is
    $$W(\e)=\begin{bmatrix}
                1 & 1 & \dots & 1\\
                (D_2(\e))_{1,1} & (D_2(\e))_{2,2} & \dots & (D_2(\e))_{n,n}\\
                (D_3(\e))_{1,1} & (D_3(\e))_{2,2} & \dots & (D_3(\e))_{n,n}\\
            \end{bmatrix}=\begin{bmatrix}
        1 & 1 & \dots & 1\\
        p(\lb_1(\e)) & p(\lb_2(\e)) & \dots & p(\lb_n(\e))\\
        \lb_1(\e) & \lb_2(\e) & \dots & \lb_n(\e)
    \end{bmatrix}$$

    For a rational function $q(x):=f(x)/g(x)$, $f,g$ coprime, let $den_x(q):= deg_x(g)$. From the work so far, we deduced that $\nu_\e(M_{i,j})\geq -k(n-1)$ for all $i,j\in [n]$, for the choices of matrix $M=U(\e)$ and $M=V(\e)$. Additionally, the vectors $w_i(\e)$ are of the form $(1, p(\lb_i(\e)),\lb_i(\e))$ and so $den_\e((w_i(\e))_j)=0$ for $j\in\{1,3\}$ and $den_\e((w_i(\e))_2)=den_\e(p(\lb_i(\e)))\leq n(n-1)$ because the coefficients of the polynomial $p\in\bC(\e)[x]$ have denominator degree at most $n(n-1)$. 
    Therefore from \cref{lem: border_bound}, $\bdorder(T)$, is bounded above by
    \begin{align*}
        \bdorder(T)&\leq min_{i,j}(\nu_\e(u_i(\e)_j))+min_{i,j}(\nu_\e(v_i(\e)_j)) + min_{i,j}(\nu_\e(w_i(\e)_j))\\
        &\leq 2k(n-1) + n(n-1)\\
        &\leq 2n(n-1)^2 + (n+2)(n-1) = 2(n-1)^3 + 3(n-1)^2 + 3(n-1) \qedhere
    \end{align*}
%\PK{Define formally the "classical error degree" before the theorem, and add this result in the Theorem's statement. Also let's try to find a better name than "classical error degree".%perhaps "approximation degree"?
% [BCS97] has the name "degeneration of order $q$ (of the unit tensor)."}
 \end{proof}

\begin{remark}\label{rem:weyr}
 The Weyr normal form is only needed for the second part of Theorem~\ref{thm: not_1-regular} (on the degeneration order). As pointed out in the proof, for the bound on the degree of error we do not need to make sure that $R$ and $B$ are upper triangular. As a result, the Weyr normal form can be replaced by the Jordan normal form for this part of the theorem.
\end{remark}
 
% %================================================================================
% \subsection{Undercomplete decomposition: the case of border rank \texorpdfstring{$r \leq \min(m,n)$}{r less than min(m,n)}}
% \label{sec:undercomplete}
% %================================================================================

% Here we extend some of the main results of Section~\ref{sec:border} to the case of a tensor of format $m \times n \times p$ and border rank $r\leq \min(m,n)$. 
% In particular, we obtain in Theorem~\ref{thm: linear part m not n} a linear perturbation under a 1-regularity assumption (thereby generalizing Theorem~\ref{thm:border1reg}).
% And we obtain in Theorem~\ref{thm: general m not n} a cubic upper bound on the degree of error for tensors with 3 slices, thereby generalizing Theorem~\ref{thm: not_1-regular}.

As a corollary, we obtain our first main theorem, on the border degeneration degree as well as error degree bounds for the undercomplete case.

\approximateMain*

\begin{proof}
    From \cref{thm: not_1-regular}, we can say that $\Brank(Z)=r$ and there exists a $Z(\e)$ converging to $Z$ with entries linear in $\e$. Let $T(\e)= (A \tp B^t \tp I)(Z(\e))$ again. Then, from \cref{cor: undercomplete decompositions}, $T(\e)$ is a tensor of rank $r$ for all $\e\neq 0$ and with entries polynomial in $\e$ with degree at most $(r-1)^3 + 2(r-1)^2 + (r-1)^2$ and $lim_{\e\rightarrow 0}T(\e) = T$. Furthermore, the degree of border degeneration is at most $2(r-1)^3 + 3(r-1)^2 + 3(r-1)$.
\end{proof}

\subsection{Overcomplete decomposition: the case  \texorpdfstring{$r \geq \min(m,n)$}{rgeqn}} \label{sec:over}

In this section we give bounds on the degree of error and on the order of border degeneration in the overcomplete setting
 (border rank $r \geq n$ for tensors of format $n \times n \times p$).
We begin with a characterization of border rank by commuting extensions, in the same style as the characterization of tensor rank in~\cite{koi24overcomplete}. First, we need to introduce a few notations.

Let $C(r,p)$ be the variety of commuting tuples $(Z_1,\ldots,Z_p)$ of matrices of size $r$. We denote by $D(r,p) \subseteq C(r,p)$ the set of tuples of matrices $(Z_1,\ldots,Z_p)$ of size $r$ that are simultaneously diagonalizable. 
Also, denote by $\pi_{r,p,n}$ the projection on the top left blocks of these matrices: more precisely, 
 $\pi_{r,p,n}(Z_1,\ldots,Z_p)=(A_1,\ldots,A_p)$ where each $A_i$ is the top left block of $Z_i$ of size $n$.
 In other words, $(Z_1,\ldots,Z_p)$ is a diagonalizable commuting extension of $(A_1,\ldots,A_p)$.
\begin{theorem} \label{th:bordercommute}
Let $T \in {\bC}^{n \times n \times p}$ be a tensor with a first slice $T_1$ that is invertible.
For any $r \geq n$, the three following properties are equivalent:
\begin{itemize}
\item[(i)] $T$ has border rank at most $r$.
\item[(ii)] The $p$-tuple $(I_n,T_2T_1^{-1},\ldots,T_pT_1^{-1})$ belongs to  $\overline{\pi_{r,p,n}(D(r,p))}$.
\item[(iii)] The $(p-1)$-tuple $(T_2T_1^{-1},\ldots,T_pT_1^{-1})$ belongs to  $\overline{\pi_{r,p-1,n}(D(r,p-1))}$.
\end{itemize}
\end{theorem}
\begin{proof}
Since the equivalence of (ii) and (iii) is clear, we focus on the equivalence between (i) and (ii) in this proof.
We assume without loss of generality that $T_1=I_n$ since multiplying all slices by $T_1^{-1}$ does not change the border rank (\cref{lem:mult}).

Consider first the case where $(I_n,T_2T_1^{-1},\ldots,T_pT_1^{-1}) = (I_n,T_2,\ldots,T_p)$ 
 belongs to the projection $\pi_{r,p,n}(D(r,p))$.
Let $(I_r,Z_2,\ldots,Z_p)$ be a diagonalizable commuting extension of size $r$ of this tuple.
The tensor $T'$ with slices  $(I_r,Z_2,\ldots,Z_p)$ has rank $r$ 
by \cref{cor: rank_n_char}.
Since $T$ is a subtensor of $T'$, we conclude that $T$ has rank at most $r$.
It follows immediately that (ii) implies (i). Indeed, if (ii) holds then $T$ belongs to the closure of the set of tensors of rank at most $r$ (as we have just shown), i.e., $T$ has border rank at most $r$.

Conversely, assume that $T$ has border rank {exactly} $r$. We can write $T=\lim_{\epsilon \rightarrow 0} T(\epsilon)$, where $T(\epsilon)$ has a decomposition
$$T(\epsilon)=\sum_{i=1}^r u_i(\epsilon) \otimes v_i(\epsilon) \otimes w_i(\epsilon) .$$
For every small enough $\epsilon$, $\Rank T(\epsilon) = r$  (otherwise, $T$ would be of border rank $<r$) and 
$T_1(\epsilon)$ is invertible since $\lim_{\epsilon \rightarrow 0} T_1(\epsilon) = T_1=I_n$.
We will assume without loss of generality that the first coordinate $w_{i1}(\epsilon)$ of $w_i(\epsilon)$  is nonzero for all $i$ and all $\epsilon$ (slightly perturb  $w_{i1}(\epsilon)$ if necessary). 
We can therefore apply~\cite[Proposition 15]{koi24overcomplete} to $T(\epsilon)$ with $A=T_1(\epsilon)$ (this choice of $A$ is possible since  $w_{i1}(\epsilon) \neq 0$). %Proposition~3.2 
Proposition 15 yields a commuting diagonalizable extension $(I_r,Z_2(\epsilon),\ldots,Z_p(\epsilon))$ 
of size $r$ for the tuple $(I_n,T_1^{-1}(\epsilon)T_2(\epsilon),\ldots,T_1^{-1}(\epsilon)T_p(\epsilon))$. 
In other words, this tuple belongs to $\pi_{r,p,n}(D(r,p))$. We conclude that (ii) holds by taking the limit $\epsilon \rightarrow 0$.

It remain to treat the case where $\rho = \Brank T < r$. We have just shown that  
%$(I_n,T_2T_1^{-1},\ldots,T_pT_1^{-1})$
$(I_n,T_2,\ldots,T_p)$  belongs to  $\overline{\pi_{\rho,p,n}(D(\rho,p))}$. 
We can add $r-\rho$ rows and columns of 0's to the matrices of size $\rho$ in these $p$-tuples, and we conclude that
%$(I_n,T_2T_1^{-1},\ldots,T_pT_1^{-1})$ 
$(I_n,T_2,\ldots,T_p)$ also belongs to  $\overline{\pi_{r,p,n}(D(r,p))}$. 
\end{proof}

\begin{remark} \label{rem:closed}
The closure $\overline{D(r,p)}$ of  $D(r,p)$  is an irreducible component of $C(r,p)$, sometimes known as its ``principal component''~\cite{jelisiejew22}.
Note that the set $\overline{\pi_{r,p,n}(D(r,p))}$ appearing in~\cref{th:bordercommute} is equal to $\overline{\pi_{r,p,n}(\overline{D(r,p)})}$.
It is therefore of interest to determine for which values of $r$, $p$ and $n$ the projection $\pi_{r,p,n}(\overline{D(r,p)})$ is closed. In such a case, 
the set $\overline{\pi_{r,p,n}(D(r,p))}$ can be replaced by $\pi_{r,p,n}(\overline{D(r,p)})$ in Theorem~\ref{th:bordercommute}.
We will return to this question shortly (in \cref{cor:projclosed}) when treating tensors of format $n \times n \times 3$.
\end{remark}

\begin{corollary} \label{cor:closure}
Let $T \in {\bC}^{n \times n \times 3}$ be a tensor with a first slice $T_1$ that is invertible.
For any $r \geq n$, the two following properties are equivalent:
\begin{itemize}
\item[(i)] $T$ has border rank at most $r$.
\item[(ii)] The pair $(T_2T_1^{-1},T_3T_1^{-1})$ belongs to  $\overline{\pi_{r,2,n}(C(r,2))}$.
\end{itemize}
\end{corollary}
\begin{proof}
By \cref{Thm: Motzkin-Taussky}, $\overline{D(r,2)}=C(r,2)$. 
The result follows from \cref{th:bordercommute} and \cref{rem:closed}.
\end{proof}
%\PK{We now have the better bound $r-1$ instead of $(r-1)^2$ in the next theorem.}
\begin{theorem} \label{th:over1regular}
    Let $T\in \bC^{n\times n \times p}$ be a tensor with a first slice $T_1$ that is invertible.
   Assume that the $(p-1)$-tuple $(T_2T_1^{-1},\ldots,T_pT_1^{-1})$ belongs to $\pi_{r,p-1,n}(C(r,p-1))$, and let $(Z_2,\ldots,Z_p)$ be the corresponding commuting extension. % of $(T_2T_1^{-1},\ldots,T_pT_1^{-1})$.
   If $Z_2$ is 1-regular then $T$ has border rank at most $r$, and there is a line of tensors $T(\epsilon) = T + \epsilon T'$ with $T' \in \bC^{n\times n \times p}$ such that $\Rank (T(\epsilon)) \leq r$ for all but finitely many $\e$.
   %\begin{itemize}
 %  \item[(i)] $T=\lim_{\epsilon \rightarrow 0} T(\epsilon)$.
 %  \item[(ii)] For all but finitely many $\epsilon$, $\Rank (T(\epsilon)) \leq r$.
  % \end{itemize}
   In particular, if $\Brank T = r$ then $\edeg(T) \leq 1$ and $\bdorder(T) \leq r-1$. %$(r-1)^2$.
\end{theorem}
\begin{proof}
By Lemma~\ref{lem:mult}, we can assume without loss of generality that $T_1=I_n$. 
Let $Z \in \bC^{r \times r \times p}$ be the tensor with slices $I_r,Z_2,\ldots,Z_p$. Since $Z_2$ is 1-regular, by Theorem~\ref{thm:border1reg}, $\Brank Z = r$ and $Z$ has error degree at most 1. Hence there is a line of tensors  $Z(\epsilon) \in \bC^{n\times n \times p}$ such that
%  \begin{itemize} 
 %  \item[(i)]
  $Z=\lim_{\epsilon \rightarrow 0} Z(\epsilon)$, and
  % \item[(ii)] 
 $\Rank (Z(\epsilon)) \leq r$  for all but finitely $\epsilon$. We obtain
 %the conclusions of the theorem 
 the result on the degree of error by projecting this line on the $p$ top-left blocks of size $n$.\footnote{Even though $\Brank Z = r$, it might be the case that $\Brank T < r$. In this case, we cannot quite conclude that $T$ has error degree at most 1.}
  % \end{itemize}
  For the degeneration order, we apply \cref{thm:border1reg}:
  %\cref{prop: linear_part}: 
  %Theorem~\ref{thm:epsilon-order}: 
  $Z$ has degeneration order
  at most $r-1$, 
  %$(r-1)^2$, 
  and this applies also to $T$ since it is a subtensor of $Z$.
\end{proof}

\begin{theorem} \label{th:overnot1regular}
  Let $T\in {\bC}^{n\times n \times 3}$ be a tensor with a first slice $T_1$ that is invertible.
  Assume that the pair $(T_2T_1^{-1},T_3T_1^{-1})$ has a commuting extension of size $r$.
  Then $\Brank T \leq r$ and:
  \begin{itemize}
\item[(i)]  There exist tensors $X_i\in {\mathbb C}^{n\times n \times 3}$ such that $T(\epsilon):= T + \epsilon X_1 + \dots + \epsilon^kX_k$ has rank at most $r$ for all but finitely many $\epsilon$. One may take here  $k \leq (r-1)^3 + (r-1)^2.$

\item[(ii)] There are vectors  $u_i, v_i, w_i$ with polynomial entries such that $$\sum_{i=1}^r u_i(\e)\otimes v_i(\e) \otimes w_i(\e) = \e^{\ell} T + \e^{\ell+1}Q$$ for some $Q\in {\mathbb C}[\e]^{n\times n \times 3}$.  Here one may take: $$\ell \leq 2(r-1)^3 + 3(r-1)^2 + 3(r-1).$$
   \end{itemize}
   In particular, if  %$\Brank T = r$ 
   the border rank of $T$ is equal to $r$ then $T$ has degree of error at most $(r-1)^3 + (r-1)^2$ and order of degeneration 
   at most $2(r-1)^3 + 3(r-1)^2 + 3(r-1).$
\end{theorem}

\begin{proof}
We assume without loss of generality that $T_1=I_n$. Let $(Z_2,Z_3)$ be the commuting extension of $(T_2,T_3)$.
Let $Z \in \bC^{r\times r \times 3}$ be the tensor with slices $I_r,Z_2,Z_3$. 
By \cref{thm: not_1-regular}, $Z$ has border rank $r$, and degree of error at most $(r-1)^3 + (r-1)^2.$
This means that there exists $k \leq (r-1)^3 + (r-1)^2$ and tensors $Y_1,\ldots,Y_k$ such that $Z(\epsilon):= Z + \epsilon Y_1 + \dots + \epsilon^kY_k$ has rank at most $r$ for all but finitely many $\epsilon$.
The slices of the  $X_i$ in item (i) of the theorem are the top-left blocks of size $n$ of the slices of the $Y_i$.
 
Moreover, since $Z$ has degeneration order $$\ell \leq 2(r-1)^3 + 3(r-1)^2 + 3(r-1),$$ there are vectors $u'_i, v'_i, w_i$ with polynomial entries such that $$\sum_{i=1}^r u'_i(\e)\otimes v'_i(\e) \otimes w_i(\e) = \e^{\ell} Z + \e^{\ell+1}Q'$$ for some $Q'\in {\mathbb C}[\e]^{r\times r \times 3}$.  
The slices of tensor $Q$ in item (ii) are the top-left blocks of size $n$ of the slices of $Q'$, and the vectors $u_i,v_i$ are made of the first~$n$ entries of $u'_i,v'_i$.
Like in \cref{th:over1regular}, we obtain a bound on the degree of error of $T$ (and on its degeneration order) only in the case $\Brank T =r$.
\end{proof}

It would be quite interesting to find out for which values of $r$ and $n$ the projections $\pi_{r,2,n}(C(r,2))$  of pairs of commuting matrices of size $r$ on their top-left blocks of size $n$ are closed. 
Indeed, for such integers the statement of \cref{th:overnot1regular} can be strengthened and simplified:
 
\begin{corollary}\label{cor:projclosed}
  Let $T\in \bC^{n\times n \times 3}$ be a tensor with a first slice $T_1$ that is invertible.
  Let $r = \Brank T$, and suppose that the projection $\pi_{r,2,n}(C(r,2))$ is closed.
  Then,  $T$ has degree of error at most $(r-1)^3 + (r-1)^2$ and degeneration order
   at most $2(r-1)^3 + 3(r-1)^2 + 3(r-1).$
\end{corollary}

\begin{proof}
By \cref{cor:closure} the pair $(T_2T_1^{-1},T_3T_1^{-1})$ belongs to  $\overline{\pi_{r,2,n}(C(r,2))}$, and to $\pi_{r,2,n}(C(r,2))$ since we are assuming that this projection is closed.
The result follows directly from \cref{th:overnot1regular}.
\end{proof}

%=======================================================================
\bibliographystyle{alpha}
\bibliography{bib}
%=======================================================================

\end{document}